\documentclass[acmtog]{acmart}

\usepackage{booktabs}

\usepackage{nicefrac}
\usepackage{bm}
\usepackage{wrapfig}

\setcitestyle{square}

\usepackage[ruled]{algorithm2e}

\SetAlFnt{\small}
\SetAlCapFnt{\small}
\SetAlCapNameFnt{\small}
\SetAlCapHSkip{0pt}
\IncMargin{-\parindent}

\usepackage{overpic}

\acmJournal{TOG}

\authorsaddresses{}
\acmPrice{}
\acmDOI{}
\acmISBN{}
\acmVolume{}
\acmNumber{}
\acmArticle{}
\acmYear{}
\acmMonth{1}
\setcopyright{none}

\begin{document}
% Title portion
\title{Seamless Parametrization with Arbitrarily Prescribed Cones}

\author{Marcel Campen}
\affiliation{%
  \institution{Osnabr\"uck University}
  \country{Germany}}
\email{campen@uos.de}
\author{Hanxiao Shen}
\affiliation{%
  \institution{New York University}
  \country{USA}}
\author{Jiaran Zhou}
\affiliation{%
  \institution{Shandong University}
  \country{China}}
\author{Denis Zorin}
\affiliation{%
  \institution{New York University}
  \country{USA}}

\newtheorem{prop}{Proposition}[section]

\renewcommand\shortauthors{Campen et al.}

\begin{abstract}
Seamless global parametrization of surfaces is a key operation in geometry processing, e.g. for high-quality quad mesh generation. A common approach is to prescribe the parametric domain structure, in particular the locations of parametrization singularities (cones), and solve a non-convex optimization problem minimizing a distortion measure, with local injectivity imposed through either constraints or barrier terms.
In both cases, an initial valid parametrization is essential to serve as feasible starting point for obtaining 
an optimized solution.
While convexified versions of the constraints eliminate this initialization requirement, they narrow the range of solutions, causing some problem instances that actually do have a solution to become infeasible.

We demonstrate that for arbitrary given sets of topologically admissible parametric cones with prescribed curvature, a global seamless parametrization 
always exists (with the exception of one well-known case). Importantly, our proof is constructive and directly leads to a general algorithm for computing such parametrizations.  Most distinctively, this algorithm is bootstrapped with a convex  optimization problem (solving for a conformal map), in tandem with a simple linear equation system (determining a seamless modification of this map).
This initial map can then serve as valid starting point and be optimized with respect to application specific distortion measures using existing injectivity preserving methods.

\end{abstract}



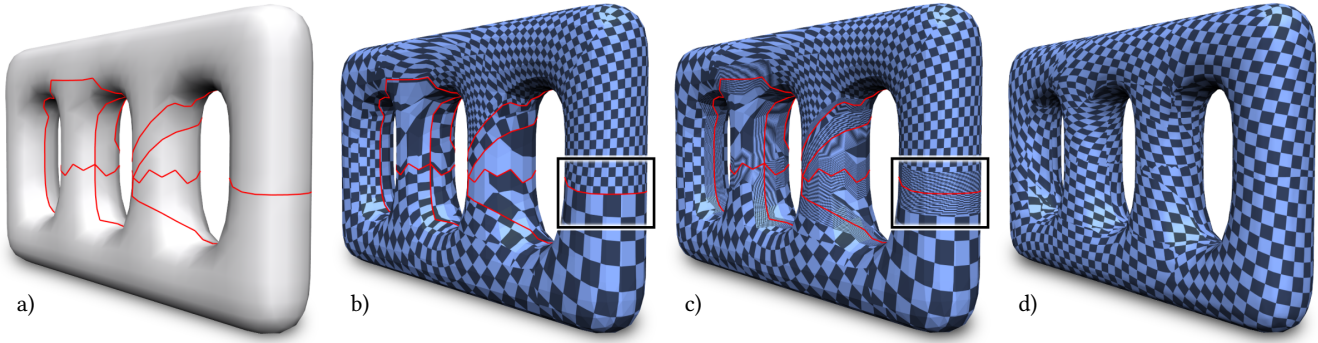
\begin{teaserfigure}
\vspace{0.1cm}
\centering
\begin{overpic}[width=0.99\linewidth]{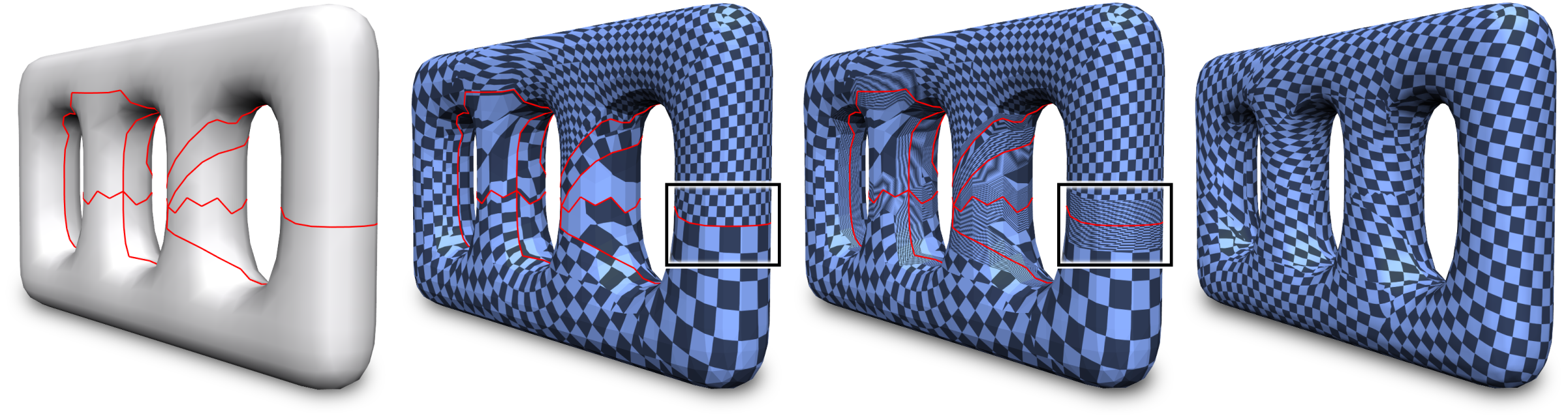}
\put(2,3.5){a)}
\put(27,3.5){b)}
\put(52,3.5){c)}
\put(77,3.5){d)}
\end{overpic}
\vspace{-0.2cm}
\caption{Overview of our method's stages: a) Cutgraph on a surface, consisting of handle loops, connectors, and one additional path. b) Conformal parametrization which maps the cutgraph's branches to axis-aligned straight segments in the parametric domain. This map is only \emph{rotationally} seamless. c) This map modified to be fully seamless by \emph{map padding}; notice that this map, though locally highly distorted, is actually seamless across the red cutgraph. d) The final map optimized for low isometric distortion, starting from the valid map in (c). Zoom-ins that more clearly expose the effect of padding are shown in Figure \ref{fig:teaserzoom}.}
\label{fig:teaser}   
\vspace{0.2cm} 
\end{teaserfigure}

\maketitle

\section{Introduction}
\label{sec:intro}

Computing global parametrizations of surfaces is a key operation in geometry processing.
While in general only disk-like surfaces can be parametrized continuously in a (locally or globally) injective manner, surfaces of arbitrary topology can be dealt with by cutting them to disks.
Across the cuts the parametrization will be discontinuous, limiting its practical value, but this is inevitable in general. 

One can, however, ask the parametric transitions across cuts to be from certain classes rather than arbitrary.
For instance, restricting to similarity transformations with a rotation by some multiple of $\nicefrac{\pi}{2}$ yields global parametrizations ideal for T-spline constructions \cite{Campen:2017:SimilarityMaps}.
Restricting further to \emph{rigid} transformations with such discrete rotation angles yields global parametrizations which (possibly after quantization \cite{Bommes:2013,CampenBK15}) are well suited for tasks like conforming quadrangulation, spline and subdivision fitting, seamless texturing, or constructing grids for solving PDEs on surfaces. This important latter type of parametrizations was termed \emph{seamless} \cite{myles2012global,Purnomo:2004}.

Seamless parametrizations have \emph{singularities}, points around which the total parametric angle is not $2\pi$ but some other multiple of $\nicefrac{\pi}{2}$, i.e.~the parametrization coordinate isolines do not locally form a regular grid. Equivalently, the metric induced by the parametrization has \emph{cones}, points where the metric is not flat, its curvature not zero but some other multiple of $\nicefrac{\pi}{2}$. Intuitively, in a quadrangulation induced by the parametrization, these singularities or cones correspond to extraordinary vertices, with valence different from~4.

As implied by the Gauss-Bonnet theorem, the total curvature of these cones is a topological invariant -- i.e.~such cones, which have a significant influence on the quality and structure of the parametrization, cannot generally be avoided.
Depending on the use case, they can be considered either an impairment or features of special interest. In either case, having the ability to control (i.e.~prescribe) them -- where they are, how many there are, what curvature they have -- is of obvious benefit. 
An important task thus is:

\vspace{0.2cm}
\emph{Compute a global seamless injective parametrization with cones exactly as prescribed (in accordance with Gauss-Bonnet).}
\vspace{0.2cm}

Probably closest to a general reliable solution to this problem is an approach by Myles et al.~\shortcite{Myles:2014}: a valid global seamless injective parametrization is guaranteed, cone preservation is aimed for but not guaranteed -- in some cases
unnecessary additional cones arise.

That these are truly unnecessary follows from the fact that the above task is actually feasible:
the \emph{existence} of such parametrizations 
follows from a theorem on existence of quad meshes with prescribed extraordinary vertices \cite{jucovivc1973theorem}. The proof is relatively complex and purely combinatorial, thus does not translate into a practical parametrization construction.

In this paper, we provide a \emph{constructive} proof for the existence of seamless surface parametrizations, that is conceptually simpler and translates to a parametrization algorithm. 
Precisely, we show:

\begin{theorem}
\label{th:main}
Given a closed surface $M$ of genus $g$ and an admissible set $C$ of cones $c_i$,
each given by a point $p_i\in M$ with a prescribed curvature value $\hat\Theta_i = (4-k_i) \frac{\pi}{2}$, $k_i\!\in\!\mathbb{N}^{>1}$, there exists a global parametrization of $M$ with cones $C$ that has seamless transitions.\footnote{We assume $k_i > 1$ for brevity, as cones with curvature $3\nicefrac{\pi}{2}$, corresponding to valence~1 vertices in a quadrangulation, 
are of low relevance in common applications; with some additional special case handling, our method could be extended to $k_i = 1$.}\\
\end{theorem}
\vspace{-0.3cm}

The terms used in the theorem are defined precisely in Sec.~\ref{sec:construction}.

A set of cones $C = \{(p_i,k_i)\}$ is called \emph{admissible} if it satisfies $\sum_{i}(1- \frac{1}{4}k_i) \!=\! 2-2g$ (Gauss-Bonnet) and if $\bm{k} \neq (3,5)$ (which is the single one notorious infeasible case \cite{jucovivc1973theorem}).

\begin{figure}[t]
\centering
\begin{overpic}[width=0.99\linewidth]{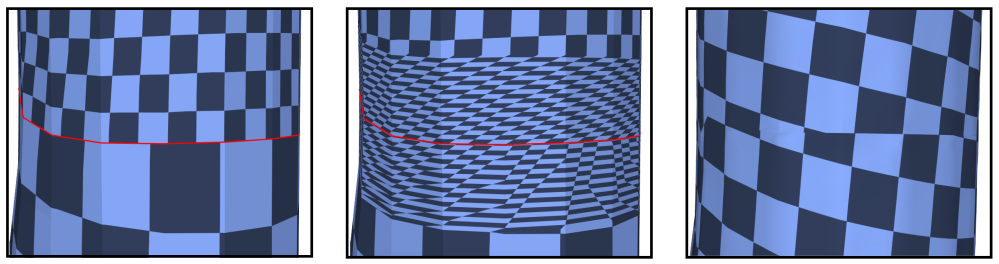}
\end{overpic}
\vspace{-0.15cm}
\caption{Zoom-ins of Figure \ref{fig:teaser}. Left: cut-aligned conformal map. Middle: padded map, with high distortion, but seamless and locally injective. Right: map optimized for low isometric distortion.}
\label{fig:teaserzoom}    
\vspace{-0.35cm}
\end{figure}

\subsection*{Basic Idea}
\label{sec:basicidea}

Instead of directly aiming for a seamless cone metric on a surface~$M$,
\begin{enumerate}
\item we cut $M$ open using a cutgraph $G$, obtaining the cut surface $M'$ consisting of one or more topological disks, 
\item ask for a cone metric on $M'$ -- \emph{without} any seamlessness requirements, but with prescribed boundary curvature; concretely, we prescribe a \emph{rectilinear} boundary, consisting of geodesically straight segments meeting at right angles,
\item modify this metric into a seamless one on $M$, yielding a seamless parametrization with exactly the prescribed cones; exploiting the rectilinear boundary property, this modification is performed by \emph{padding} the straight segments in the parametric domain with rectangles of suitably chosen size.
\end{enumerate}
The metric in (2) is known to exist; e.g.~a \emph{conformal} metric with prescribed cones and boundary curvature (satisfying Gauss-Bonnet) on a disk always exists -- in the smooth setting (cf.\ Sec.~\ref{sec:conemetrix}); the situation is more complicated in the discrete setting (cf.~Sec.~\ref{sec:62}). Figures~\ref{fig:teaser} and~\ref{fig:teaserzoom} show the outcome of these main steps. We refer to Appendix \ref{sec:illexample} for a comprehensive example illustrating these steps in a concrete, simple case.

\subsection*{Key Contributions}

Our key technical contributions pertain to step (3) of the above outline:
\begin{itemize}
\item We propose a technique (\emph{map padding}) to modify the non-seamless metric into a seamless one on $M$.
\item We prove that, for certain choices of cutgraph combinatorics, this technique is always applicable and succeeds.
\item We describe an implementation of this construction for the discrete, piecewise linear case.
\end{itemize}

So, in essence, our approach turns on a problem reduction:

\vspace{0.2cm}
\hrule
\vspace{0.13cm}
 {If one is able to compute a metric with prescribed cones and prescribed boundary curvature for \emph{disk-topology surfaces}, this solves (by means of our technique) the more general problem of computing a global seamless parameterization with prescribed cones for \emph{arbitrary-topology surfaces.}
 \vspace{0.13cm}
\hrule
 \vspace{0.05cm}
 
  \vspace{0.17cm}
 Our algorithm can thus be used to obtain non-degenerate, locally injective, seamless parametrizations of arbitrary closed discrete surfaces (triangle meshes) with arbitrary cones, assuming the initial metric with prescribed boundary curvature can be obtained.

\section{Related Work}
\label{sec:related}
Seamless surface parametrization and the related subject of quadrangulation and quad layout generation is a well-explored topic.
A relatively recent survey \cite{bommes2013quad} has references to many works in this 
area. We focus here on the most closely related ones.

In a wide variety of applications, surface parametrizations are required to be (locally) injective (i.e.~without fold-overs) as well as to exhibit low parametric distortion \cite{floater2005spt}.
Due to the challenging nature of this requirement, a common strategy is to proceed in a two-step fashion:
first construct an initial injective parametrization (without specific attention to distortion), then optimize it with respect to application specific distortion criteria (while preserving injectivity). 
Our work likewise follows this strategy.

\paragraph{Constructing injective maps} Interestingly, whenever a robust process is desired, injective maps are almost always initialized using one classical result on convex harmonic maps \cite{Tutte,Floater1997} (essentially a discrete version of the Rad\'o-Kneser-Choquet theorem).
In its original form,
it handles surfaces with disk topology and does not support cones.
Some recent results \cite{GORTLER200683,Aigerman:2015,Aigerman:2016:HOT,Bright:2017:HGP} elegantly generalize the idea to other settings, but either not to arbitrary sets of cones, not to arbitrary topology, not using the piecewise linear Euclidean setting, or without similar guarantees on map existence. 

\paragraph{Injectivity-preserving optimization} A variety of techniques have been presented for distortion optimization, e.g.\ \cite{Schueller:LIM:2013,Hormann:2000:MAE,Rabinovich:2017:SLI,Kovalsky:2016:AQP,Zhu:2018:BCQ,Liu:2018:PP,Shtengel:2017:GOV}. Through line search techniques, barrier functions, and similar techniques they are able to guarantee preservation of injectivity -- if initialized with an injective starting point.
State-of-the-art techniques can handle large meshes efficiently and tolerate significant 
imperfections in the initial solution. 

\paragraph{Seamless parametrization}
A number of methods have been described for the construction of seamless parametrizations with prescribed cones \cite{kalberer2007qsp,Bommes:2009,Bommes:2013,myles2012global,myles2013controlled,esck2016,Bright:2017:HGP,Fu:2015:CLI,Chien:2016:BDP}.
Interestingly, but not surprisingly, they do not follow the above two step principle -- as no general method for the first step (valid initialization) is known for the arbitrary-topology arbitrary-cones setting. 
Instead, they are typically based on optimization subject to non-convex constraints and, despite long development and practical importance, no concise sufficient conditions for success are known.
The key issue is that there is no available way  to construct an initial solution, and one cannot guarantee that the solver will itself find a way into the feasible region.

Only for certain special cases there are known solutions in this regard, e.g.,\ for specific genus or specific cones \cite{Aigerman:2015,Gu:2003}, using more general non-piecewise-linear parametrization \cite{Aigerman:2016:HOT}, or requiring additional input \cite{tong2006dqd}. Particular challenges are caused by the fact that the given surface discretization may not even admit a (elementwise linear) solution, i.e.\ systematic remeshing capabilities are required in any approach that is supposed to be reliable.

\paragraph{Quadrangulation}

The problem of surface quadrangulation with conforming elements and prescribed extraordinary vertices is closely related -- state-of-the-art methods actually construct quadrangulations via seamless parametrization \cite{bommes2013quad}.
\cite{jucovivc1973theorem} investigate the question of existence of such quadrangulations. The result is purely combinatorial and does not yield a surface parametrization. On an abstract level, we adapt some of the general ideas in this work as foundation of our approach to modify non-seamless into seamless parametrizations through map padding.

In this context of quadrangulation, our strategy of transitioning from an initial non-seamless parametrization to a seamless one is, in a sense, similar to modifying a non-conforming quadrangulation into a conforming one.
This has been tackled by simple subdivision or more involved T-mesh simplification techniques \cite{Myles:2014} -- however, at the expense of not always preserving the prescribed extraordinary vertices. Our modification technique, by contrast, always preserves exactly the prescribed cones.

\paragraph{Cone selection}
Regarding the (application specific) question of which cones to prescribe, common approaches are based on considering surface curvature (e.g.\ via cross fields \cite{Vaxman:FieldsSTAR}), distortion reduction \cite{Kharevych:2006:DCM,Soliman:2018:OCS,BenChen:2008,egsh2017}, or on manual design and editing of, e.g., quad mesh or quad layout structure \cite{esck2016}.

 \paragraph{General holonomy prescription.} 
\cite{Campen:2017:SimilarityMaps} address a related problem, showing that for any admissible \emph{holonomy signature} one can construct (also via conformal maps) a \emph{seamless similarity} map adequate for constructing T-splines.
A holonomy signature, in addition to prescribed cone angles, includes turning angles around homology loops.
In contrast, we use a stronger notion of seamlessness, not allowing scale jumps across cuts, while not controlling global turning angles around homology loops (cf.~Sec.~\ref{sec:conclusion}) -- however, they are of the form $k\nicefrac{\pi}{2}$ (for \emph{some} $k$) by our construction.

\section{Seamless Parametrization Construction}
\label{sec:construction}

First, we define the seamless parametrizations we aim to construct as set out by Theorem~\ref{th:main}, as well as a weaker notion of
a \emph{rotationally seamless} parametrization we need as an intermediate step.

  Suppose a smooth surface $M$ is cut to a set of topological disks $M^c_i$ by a cutgraph $G$, i.e. a collection of smooth curves (\emph{branches}) $\gamma_j$ embeddeded in $M$ meeting only at their endpoints (\emph{nodes}).  We call the resulting cut surface  $M^c$; the boundary of $M^c$ consists of curves  $\gamma^c_j$ (boundary curves).  There is a canonical map $\pi: M^c \rightarrow M$, which is identity in the interior of $M^c$ and maps exactly two boundary curves $\gamma^c_j$ to each branch $\gamma_j$ on $M$.  Pairs of boundary curves mapping to the same branch $\gamma_j$ are called \emph{mates}, and boundary points where curves $\gamma^c_j$ meet are called \emph{joints}.
    The image of any joint under $\pi$ is a node.
Pairs of non-joint points $p, q \in \partial M^c$ with $\pi(p) = \pi(q)$ are called \emph{mated points}.
For a boundary point $p$ let $t_p  \in T_pM^c$ denote a unit vector that is tangent to the boundary $\partial M^c$ at $p$.

\begin{definition}[Rotationally Seamless Parametrization]
A continuous, locally injective map $F: M^c \rightarrow \mathbb{R}^2$ is called \emph{rotationally seamless} parametrization of $M$, if
for any pair $p$, $q$ of mated points, the images $dF_p(t_p)$ and $dF_q(t_q)$ of boundary tangents
are related by a similarity transformation $T_{pq}$, i.e. $T_{pq}\!\circ\! dF_p(t_p) = dF_q(t_q)$, with a rotation angle that is a multiple of $\nicefrac{\pi}{2}$
and constant per branch.
\end{definition}  
Such a rotationally seamless parametrization does, in general, have a (pointwise) \emph{scale jump} across the cut (cf.~Fig.~\ref{fig:rotseam} left)
-- unless the similarity $T_{pq}$ is actually just a rotation everywhere:

\begin{definition}[Seamless Parametrization]
 A map $F: M^c \rightarrow \mathbb{R}^2$ is called a \emph{seamless} parametrization of $M$, if
 it is rotationally seamless and for each pair of mated points $p$, $q$ the transition $T_{pq}$ is rigid, i.e.
it is a rotation with a rotation angle that is a multiple of $\nicefrac{\pi}{2}$.
\end{definition}

\begin{figure}[b]
\centering
\begin{overpic}[width=0.99\columnwidth]{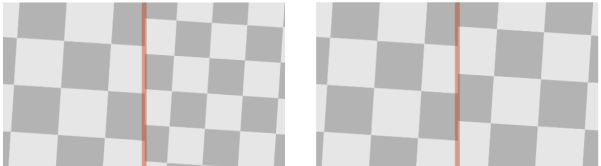}
\end{overpic}
\vspace{-0.3cm}
\caption{Visualization of a parametrization on a surface near a cut branch~(red). Left: rotationally seamless. Right: seamless.}
\label{fig:rotseam}    
\vspace{-0.0cm}
\end{figure}  

Notice that seamlessness implies that the images $F(\gamma^c_j)$ and $F(\gamma^c_k)$ of mates $\gamma^c_j$ and $\gamma^c_k$ are congruent.

A seamless parametrization induces a metric on the surface $M$ which is flat except at the nodes, where it may be singular; it may have a \emph{cone}. We say that a seamless parametrization has a cone with angle $\alpha$ at a node $p$, if the sum of parametric angles at all joints $q$ in $M^c$ with $\pi(q) = p$ is equal to $\alpha$. This cone has curvature $\Theta = 2\pi - \alpha$.

\paragraph{Overall Approach} We first construct a parametrization $F$ that is rotationally seamless, using a specific type of (conformal) maps: maps with rectilinear boundary, i.e. with the image of the boundary of the cut surface consisting of straight segments meeting at right angles. Then this parametrization is modified near the boundary to make the scale jump vanish so as to make it into a seamless parametrization $F^s$. This is done using a process we call \emph{map padding}.
Key to our construction is cutting the surface into \emph{two} (in special cases three or four) topological disks, using a particularly structured cutgraph. This
is critical for our method of converting rotationally seamless parametrizations into seamless parametrizations.

\subsection{Cutting to Disk(s)}
\label{sec:cg}
We construct the graph $G$ needed to define our  parametrizations in two steps, first cutting the surface $M$ into a set of topological disks $M'_i$, forming a surface $M'$. Typically we use two disks, with some exceptions for special genus 2 cases.
$M'$ contains all cones in the interior. The final cut surface $M^c$ is obtained by adding and cutting along additional branches passing through all cones, without splitting the disks $M_i'$. This second step is explained in Section~\ref{sec:mtop}.

For the first step, we consider a particular type of cutgraphs that only have nodes of degree 4 and 3.
Pairs of cyclically sequential branches  around nodes form sectors: four at degree 4 nodes, three
\begin{wrapfigure}{r}{0.25\linewidth}
  \vspace{-0.35cm}
  \hspace{-0.75cm}
    \begin{overpic}[width=1.3\linewidth]{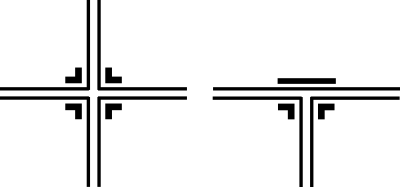}

       \put(70,31){\small flat}

    \end{overpic}
  \vspace{-0.8cm}
\end{wrapfigure}
at degree 3 nodes.
At degree 4 nodes, all four sectors are marked as \emph{corners} (cf.~Fig.~\ref{fig:cutcorners}).
At degree 3 nodes, two sectors are marked as corners, the third one is referred to as \emph{flat}. We refer to degree 3 nodes as \emph{T-nodes}.
  
  We denote the boundary curves of $M'$ by  $\gamma'_j$. 
  Any pair of sequential boundary curves of $M'$ corresponds to a corner or a flat joint.
 As we will require boundary curves to be straight and corners to be right-angled under a certain metric in the following,
 the number $n_i$ of corners on the boundary of each connected component $M'_i$ needs to match the total prescribed cone curvature in the interior of $M'_i$ as per Gauss-Bonnet, i.e.,
\begin{equation}
\label{eq:gaussdisk}
 n_i \frac{\pi}{2}  +\!\!\!\!\! \sum_{(p_j,\hat\Theta_j)\in\, C_i'} \!\!\!\!\!\!\hat\Theta_j = 2\pi,
\end{equation}
where $C_i'\subseteq C$ is the subset of cones prescribed within $M'_i$.
Note that this is equivalent to $n_i = 4+ \sum_{(p_j,\hat\Theta_j)\in\, C_i'}(k_j-4)$.

\begin{definition}[Admissible Cutgraph]
\label{def:admcut}
 A cutgraph with corner marking is \emph{admissible}, if
  \begin{itemize}
  \item all branches are embedded smooth curves meeting transversally at nodes of degree 3 or 4, and not
    passing through cones; 
  \item it partitions the surface into disk-topology components;
  \item the number of corners of each component satisfies  Eq.~\eqref{eq:gaussdisk};
  \item if a boundary curve is involved in a flat joint, its mate is not.
  \end{itemize}
\end{definition}

\begin{figure}[t]
\centering
\begin{overpic}[width=0.83\columnwidth]{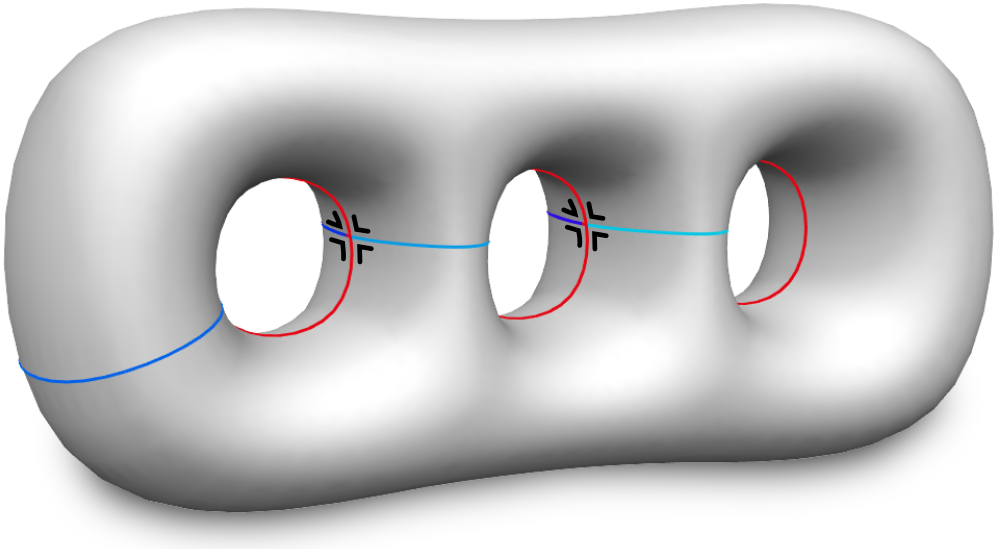}
\end{overpic}
\vspace{-0.3cm}
\caption{Degree 4 cutgraph on a surface of genus $g = 3$. This cutgraph has $10$ branches and $5$ degree 4 nodes, thus 20 corners (marked black). The cutgraph consists of loops (red) and connectors (shades of blue) (cf.~Sec.~\ref{sec:holechain})}
\label{fig:cutcorners}    
\vspace{-0.2cm}
\end{figure}  

{\subsection{Cone Metric with Rectilinear Boundary} \label{sec:conemetrix}

  Corners partition the boundary $\partial M'$ into \emph{segments}. Note that a segment may contain flat joints, thus consist of several boundary curves $\gamma'_j$ (\emph{complex segment}). 
  (All cutgraphs we will be working with contain at most two T-nodes, thus two flat joints, i.e. almost all segments are simple segments.)

  We now require a cone metric on $M'$ which has a \emph{rectilinear} boundary: under such a metric segments are geodesically straight (i.e. there is zero geodesic boundary curvature along $\partial M'$ in the interior of segments), and sequential segments form right inner angles of~$\nicefrac{\pi}{2}$.

\begin{prop}
\label{prop:conemetric}
On a cut surface $M'$, obtained from a smooth surface $M$ by cutting it along an admissible cutgraph $G$, there is a cone metric with rectilinear boundary  and prescribed admissible cones $C = \{(p_i,\hat\Theta_i)\}$.
\end{prop}
The proposition is proven in Sec.~\ref{sec:troyaproof}. In particular, a \emph{conformal} cone metric with these properties exists; conformality, however, is not essential in the following.

Note that the metric angle on $M$ around points on $G$ is $2\pi$ everywhere: points in the interior of branches are surrounded by two sectors with angles $\pi+\pi$, node points are surrounded by three sectors with angles $\pi+2\frac{\pi}{2}$ or four sectors with angles $4\frac{\pi}{2}$. As we preserve these angles in the following, this implies that no spurious cones emerge.

\begin{figure}[b]
\centering
\begin{overpic}[width=0.99\columnwidth]{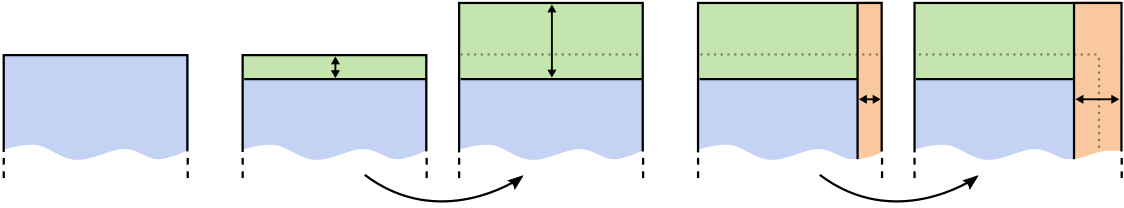}
\put(8,-1){\small a)}
\put(29,-1){\small b)}
\put(48,-1){\small c)}
\put(69.5,-1){\small d)}
\put(89.5,-1){\small e)}
\end{overpic}
\vspace{-0.1cm}
\caption{a) Generic local view of the boundary of map $F(M^c)$, with straight segments and right-angled corners. b) A rectangular strip along a segment is marked. c) The strip is stretched outwards, effectively increasing the length of the two  adjacent segments left and right of the central segment. d) This padding operation can be applied in sequence to further segments.
}
\label{fig:padrect}    
\end{figure}
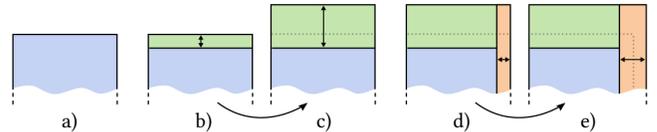

{\subsection{Metric to Rotationally Seamless Parametrization} 
  \label{sec:mtop} The cone metric with rectilinear boundary is flat away from the cones on $M'$.
  We now extend the cutgraph $G$ by a set of trees $T_i$, yielding the extended cutgraph $G^T = G \cup T$,
   where $T$ is the union of trees $T_i$. 
   The tree $T_i$ is rooted on $\partial M'_i$ at a single non-joint point, and its branches connect all cones prescribed within $M_i'$.
  Let $M^c_i$ be the surface obtained by cutting $M'_i$ along $T_i$.
  For distinction, the boundary curves of $M^c$ are denoted $\gamma_j'$ if they map to branches of $G$, or $\gamma_j^T$ if they map to branches of $T$.
   The above constructed cone metric is flat in the interior of $M^c_i$ (as the cones lie on $\partial M^c$), and 
   defines (via integration) a map $F_i: M^c_i \rightarrow \mathbb{R}^2$. It is unique up to a rigid transformation, which we choose such that all segments' images are axis-aligned in $\mathbb{R}^2$ -- which is possible because they (due to rectilinearity) are all straight and meet at right angles.
   Together, these maps $F_i$ define a global parametrization $F$ of $M$.

\begin{prop}
\label{prop:rotseam}
The  map $F$ is a rotationally seamless parametrization of $M$ (but not, in general, seamless).
\end{prop}

\begin{proof}
Due to all segments' images being axis-aligned, the angle between the images of any two mated boundary curves $\gamma'_j$, $\gamma'_k$ is some multiple of $\nicefrac{\pi}{2}$, constant per branch.
The images of any two mated boundary curves $\gamma^T_j$, $\gamma^T_k$ are congruent (thus in particular similar) as the metric is flat on $T$ by construction, and the rotation between them is a multiple of $\nicefrac{\pi}{2}$ because the prescribed angles at cones are multiples of $\nicefrac{\pi}{2}$ (cf., e.g., \cite{Springborn:2008}). Hence, $F$ is seamless on $T$ but, in general, only rotationally seamless on~$G$.
\end{proof}

\label{sec:vis}
      {\paragraph{Visualization} For purposes of illustration, we would like to visualize the image $F(M^c)$. Due to global overlaps implied by negative curvature cones,
      this is not an easy task. However, locally, near the cutgraph $G$, $F(M^c)$ always looks like in Fig. \ref{fig:padrect}a -- because the boundary consists exclusively of straight segments meeting at right angled corners (the only exception being the one boundary curve per $M_i$ where the tree $T_i$ is rooted). We use this type of illustration when a local view is sufficient. An alternative is to flatten the surface to the plane without cutting to the cones (using the trees $T$), instead (for visualization purposes) pushing the curvature of the cones evenly onto the boundary $\partial M'$.  This leads to a flattening of $M'_i$ as shown in Fig.~\ref{fig:padABC2}a, where straight boundary segments appear as curved arcs  (and cones are not visible). This makes it possible to visualize the complete rectilinear boundary without cuts
        or overlaps \label{sec:vis}}

\begin{figure}[bt]
\centering
\begin{overpic}[width=0.99\columnwidth]{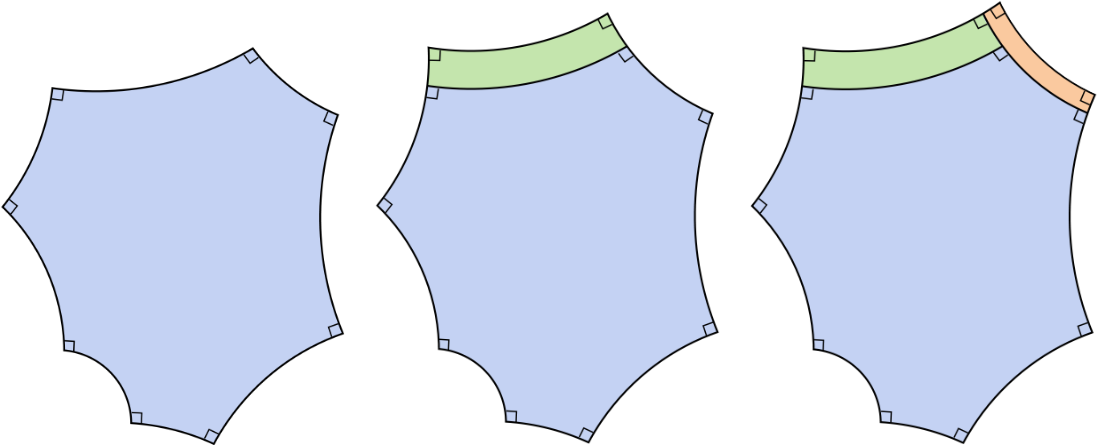}
\put(4,0){\small a)}
\put(38.5,0){\small b)}
\put(72,0){\small c)}

\put(4.5,25){\small 0}
\put(14.5,29.6){\small 1}
\put(24,30){\small 2}
\put(26.5,19.5){\small 3}
\put(22,6.9){\small 4}
\put(15.7,2.8){\small 5}
\put(11,7){\small 6}
\put(5.9,15.5){\small 7}
\end{overpic}
\vspace{-0.15cm}
\caption{a) Global visualization (without cuts to cones) of the rectilinear map, where straight segments appear as curved arcs (as explained in~Sec.~\ref{sec:vis}).
b) Padding (analogous to Fig.~\ref{fig:padrect}) of segment 1, increasing the lengths of segments 0 and 2. 
c) Padding of segment 2, increasing the lengths of segments 1 and 3. This can be continued to adjust all segments' lengths.
}
\label{fig:padABC2}    
\vspace{-0.25cm}
\end{figure}
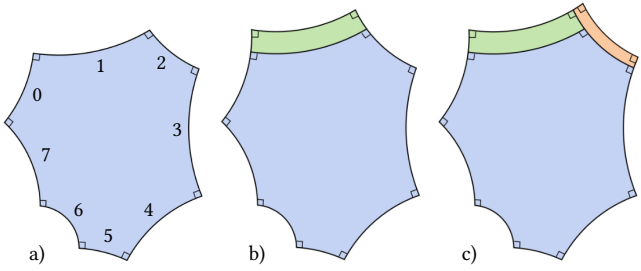  

\subsection{Seamless Parametrization by Padding} 
\label{sec:pad1}

The rotationally seamless map $F$ falls short of being seamless on two levels:
the images of mated segments may have different lengths, which implies a scale jump; but even if they are of equal length, this only implies that the scale is equal on average rather than pointwise along the corresponding branch.

We thus modify $F$ by composing it with two types of local segment-wise maps:
\begin{itemize}
\item a \emph{stretch} map $g_j$ which effects a change of the lengths of segment images,
\item a \emph{shift} map $r_j$ which subsequently equidistributes scale along a segment.
\end{itemize}

We apply these operations per segment in an iterative manner.

\begin{figure}[b]
\vspace{-0.15cm}
\centering
\begin{overpic}[width=0.99\columnwidth]{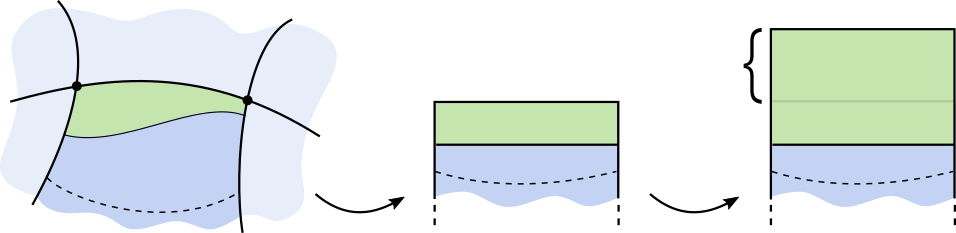}
\put(16,19){\small $M^c$}
\put(27.6,16){\small $G$}
\put(10,12.3){\small $S_j$}
\put(15.5,13.7){\small $s_j$}
\put(14.5,8){\small $\Omega_j$}
\put(37,4){\small $F$}
\put(53,10.8){\small $R_j$}
\put(71.4,4.5){\small $g_j$}
\put(88.6,14.5){\small $R_j'$}
\put(72.8,16.5){\small $w_j$}
\end{overpic}
\vspace{-0.15cm}
\caption{Illustration of strip definition and stretch map applied to perform padding of a segment $s_j$ by padding width $w_j$, cf.~Sec.~\ref{sec:pad1}.
}
\label{fig:padR}    
\end{figure}
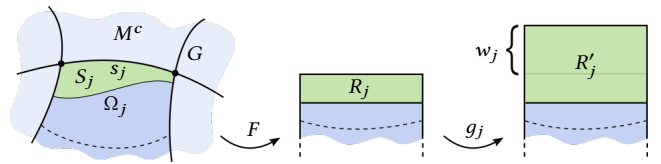

\paragraph{Stretching}
For a boundary segment $s_j$, we consider a thin strip $S_j$ on $M^c$ which runs along the entire segment and maps  to a rectangular region $R_j$ via $F$. This is illustrated in Fig.~\ref{fig:padrect}b and Fig.~\ref{fig:padR}.

More formally, the strips are defined as follows. 
The restriction of $F$ to the segment $s_j$ (which maps $s_j$ to a straight segment in the plane) is bijective and so is the restriction $F_j$ to a sufficiently small neighborhood $\Omega_j \subset M^c$ of $s_j$.
We choose the rectangle $R_j$ within $F(\Omega_j)$ such that it includes $F(s_j)$ but no cone points and no joints except the ones on $s_j$. The thin strip $S_j$ on $M^c$ is then defined as $S_j = F_j^{-1}(R_j)$, as shown in Fig.~\ref{fig:padR}.

Outside of $S_j$ we preserve the map, but within $S_j$ we modify it by a one-dimensional scaling $g_j$ such that $S_j$ is mapped onto a larger rectangle $R_j'\supseteq R_j$ whose width (orthogonal to the segment $s_j$) is increased by a \emph{padding width} $w_j$ such that it extends across the original segment image by that width. This is illustrated in Fig.~\ref{fig:padrect}c and Fig.~\ref{fig:padR}. Effectively, the domain is locally \emph{padded} by an additional rectangular region $R_j' \backslash R_j$ of width $w_j$ along the image of $s_j$, cf.~Fig.~\ref{fig:padABC2}.

The situation is slightly different at the one segment $s_j$ per component where $T$ is rooted (cf. Sec.~\ref{sec:mtop}): it is separated into two parts by $T$ (cf.~Fig.~\ref{fig:poly} left).
Both parts can, however, be handled separately using the same technique, as detailed in Sec.~\ref{sec:padding}.

We define the padded map $F^p$ iteratively, iterating over the (arbitrarily ordered) strips $S_j$, $j = 1, \dots, n$, of each connected component of $M^c$.
$F^{p,0}$~coincides with $F$, and $F^{p,m+1}$ differs from $F^{p,m}$ only on $R_{m+1}$, where it is defined
as $g_{m+1} \circ F^{p,m}|_{R_{m+1}}$, where $g_{m+1}$ is the above scaling transformation (detailed in Sec.~\ref{sec:padding}). $F^p = F^{p,n}$.

\paragraph{Shifting}
In the padded map $F^p$ we again consider for each segment a strip $S^p_j$ (now defined based on $F^p$) and modify the map within this strip by a map $r_j$ with the following properties:
\begin{itemize}
\item its restriction to segment $s_j$ reparametrizes $s_j$ to constant speed, i.e. scaled arc-length, (in the case of a complex segment: \emph{piecewise}, i.e. constant per boundary curve, cf.~Sec.~\ref{sec:padding}),
\item it is identity on the rest of the strip's boundary,
\item it is continuous and bijective.
\end{itemize}

We define the shifted map $F^s$ iteratively, again iterating over the strips $S^p_j$ of each connected component of $M^c$.
$F^{s,0}$~coincides with $F^p$, and $F^{s,m+1}$ differs from $F^{s,m}$ only on $R^p_{m+1} = F^{s,m}(S^p_{m+1})$, where it is defined
as $r_{m+1} \circ F^{s,m}|_{R^p_{m+1}}$, where $r_{m+1}$ is the above shift map (detailed in Sec.~\ref{sec:padding}). $F^s = F^{s,n}$.

\begin{prop}
\label{prop:rotseam2}
$F^{s}$ is a rotationally seamless parametrization of $M$ with the same cones as $F$, and not only the angle but also the scale jump is constant per branch of the cut.
\end{prop}

A proof is given in Sec.~\ref{sec:padding}.
The choice of padding widths determines the lengths of segments under $F^p$ and, as they are preserved by the shift maps, under $F^s$.
In the following Sec.~\ref{sec:equalize} we detail how \emph{equalizing} padding widths can be found:

\begin{definition}[Equalizing Padding Widths]
A set of padding widths  leading to all pairs of mates being of equal length under $F^s$ is called \emph{equalizing padding widths}.
\end{definition}

\begin{prop}
\label{prop:seam}
If $F^{s}$ is constructed using equalizing padding widths, it is a seamless parametrization of $M$.
\end{prop}
\proof{The constant scale jump per branch ascertained by Prop. \ref{prop:rotseam2} together with the equal lengths of mated boundary curves' images implies a scale jump of zero per branch, thus rigid transitions.\qed}

{\subsection{Length Equalization} 
\label{sec:equalize}

\begin{figure}[b]
\centering
\begin{overpic}[width=0.44\columnwidth]{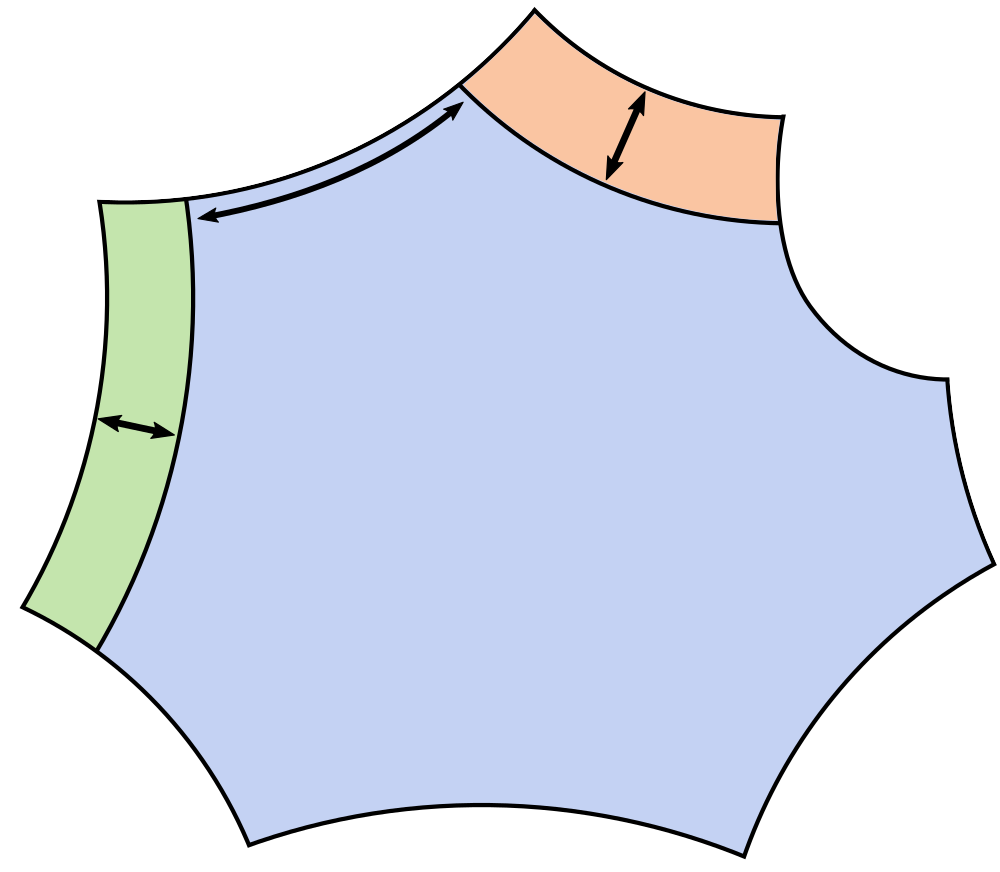}
\put(-18,45){$w_{\text{prev}(i)}$}
\put(35.5,62.5){$\ell_{i}$}
\put(62,82){$w_{\text{next}(i)}$}
\end{overpic}
\vspace{-0.05cm}
\caption{The length of segment $i$ is affected by the padding of the two adjacent segments: the original length $\ell_i$ changes to $\ell_i+w_{\text{prev}(i)} + w_{\text{next}(i)}$.}
\label{fig:padE2}    
\end{figure}  

When a boundary segment is padded, the lengths of the two adjacent segments' images change.
To make this precise, let $\ell_i$ be the length of segment $s_i$ before any padding is performed, and $w_i$ be the amount of padding applied to $s_i$. 
The length $\ell_i'$ of $s_i$ after each segment was padded according to values $\bm{w} = (w_0,w_1,w_2,...)$ is 
\begin{equation}
\ell_i' = w_{\text{prev}(i)} + \ell_i +  w_{\text{next}(i)},
\end{equation}
where prev$(i)$ and next$(i)$ are the two segments adjacent to $s_i$, preceding and following it in cyclic order along $\partial M'$, cf.~Fig.~\ref{fig:padE2}.

Our goal is to find an equalizing assignment of padding width variables $\bm{w}$ such that the lengths $\bm{\ell}' = \{\ell_0',\ell_1',\ell_2',...\}$ after padding are equal for each pair $(s_i,s_j)$ of mated segments, i.e.~ $\ell_i' =\ell_{j}'$.
This leads to \emph{length~equalization} equations 
\begin{equation}
\label{eq:equa}
w_{\text{prev}(i)} + w_{\text{next}(i)} - w_{\text{prev}(j)} - w_{\text{next}(j)}= \ell_{j} - \ell_i.\vspace{0.1cm}
\end{equation}

However, only if all cutgraph nodes are of degree 4, are all segments simple, thus mated in pairs.
When T-nodes, thus complex segments, are involved, the situation is a little different: a complex segment $s_i$ consists of multiple boundary curves; their mate curves, however, form simple segments due to the last property of Def.~\ref{def:admcut}. Hence, generally, a (simple or complex) segment $s_i$ is mated with a sequence $J_i = (s_j, s_k, ...)$ of one or more simple segments.
Length equalization equations then take this more general form:
\begin{equation}
\label{eq:equaGeneral}
\ell'_i = \sum_{j\in J_i} \ell'_j 
\vspace{-0.2cm}
\end{equation}
which expands to
\begin{equation}
\label{eq:equaGeneralw}
w_{\text{prev}(i)} \!+\! w_{\text{next}(i)}  - \sum_{j\in J_i}\! \left(w_{\text{prev}(j)} \!+\! w_{\text{next}(j)}\right) = \sum_{j\in J_i} \ell_j -\! \ell_i.
\end{equation}

These equations form a globally interdependent equation system:
\begin{equation}
\label{eq:equas}
A\bm{w} = \bm{b}, \quad w_i \geq 0 \,\forall i.
\end{equation}
Notice the non-negativity condition; it ensures that the padding operation actually stretches ($w > 0$) rather than squeezes ($w<0$) the strips along segments (which could be done only by a very limited amount).

This system needs to be solved to achieve length equalization and thus enable seamlessness.
Unfortunately, it is not generally feasible -- it may have no non-negative solution or even no solution at all. Notice that the system matrix structure is entirely determined by the cutgraph's combinatorics, leading to the following definition.
\begin{definition}[Equalizable Cutgraph]
An admissible  cutgraph for which equalization  system \eqref{eq:equas} is feasible for arbitrary  $\bm{b}$, is called \emph{equalizable}.
\end{definition}

Our key result, accompanying the proposed map padding technique, is a proof showing that there is an equalizable cutgraph for any genus and any admissible set of cones, as well as an efficient algorithm to construct such cutgraphs (cf.\ Sec.~\ref{sec:cutgraphs} and~Sec.~\ref{sec:eqproof}).

\label{sec:background}

\section{Equalizable Cutgraphs}
\label{sec:cutgraphs}

The foundation of our  construction of equalizible cutgraphs is a so-called \emph{hole chain}. 
Variations thereof, depending on the surface's genus and the cone configuration, then yield equalizable cutgraphs.
\begin{prop}
\label{th:feas1}
For any genus $g$ and any admissible prescription of cones $C$, there is an equalizable cutgraph, i.e. we can always obtain equalizing padding widths that enable a seamless parametrization.
\end{prop}

Cutgraphs covering all cases (arbitrary genus, arbitrary admissible cones) are defined in the following. Their equalizability is shown in Sec. \ref{sec:eqproof}}.

Together with Prop.~\ref{prop:conemetric} (rectilinear cone metric), Prop.~\ref{prop:rotseam} (rotationally seamless map), and Prop.~\ref{prop:seam} (seamless modification given equalizing padding widths), this Prop.~\ref{th:feas1} (equalizable cutgraphs) \textbf{concludes the constructive proof of the main theorem~\ref{th:main}.}

\subsection{Hole Chain}
\label{sec:holechain}
Given a closed surface $M$ of genus $g > 0$, we cut it along $g$ non-intersecting non-homotopic non-separating smooth loops $\alpha_i$. This yields a topological sphere $M^\circ$ with $2g$ holes.
Note that each loop corresponds to two holes, which are called \emph{partners}.
Let the holes be numbered from $0$ to $2g-1$, and denoted $h_i$, in such a way that $h_0$ and $h_{2g-1}$ (called \emph{terminals}) are partners.

Let $\pi: M^\circ \rightarrow M$ be the canonical map from $M^\circ$ to $M$,
  taking $h_i$ and its partner $h_j$ to their corresponding loop $\alpha$: $\pi(h_i) = \pi(h_j) = \alpha$.

On each hole $h_i$ pick two distinct points $q_i$ and $q_i'$, such that they are identified across partners on $M$, i.e. for partners $h_i$, $h_j$ we have $\pi(q_i) = \pi(q_j)$ and $\pi(q_i') = \pi(q_j')$.
For each $0 \leq i < 2g-1$ we then further cut $M^\circ$ along a smooth non-intersecting path between holes $h_i$ to $h_{i+1}$, starting transversally at point $q_i'$ and ending transversally at point $q_{i+1}$. These paths are called \emph{connectors}. Note that after each such cut the surface remains a topological sphere with holes (each time one less), thus is path-connected; therefore these connectors always exist.
Loops and connectors together form a cutgraph we call \emph{hole chain}, as abstractly depicted in Fig.~\ref{fig:holechain} and on a surface in Fig.~\ref{fig:cutcorners}, which yields the surface $M'$, a sphere with one hole, i.e. a disk. Loops and connectors are assumed not to cross any prescribed cone point.

\begin{prop}
The hole-chain cutgraph for any genus $g > 0$ is admissible.
\end{prop}
\proof
As the connectors' endpoints $q_i$, $q_j$ are identified in pairs on $M$ across partners $h_i$, $h_j$, each resulting cutgraph node (at point $\pi(q_i) = \pi(q_j)$ on $M$) is of degree 4 (cf.~Fig.~\ref{fig:cutcorners}).
All branches are smooth curves meeting transversally at their endpoints and not crossing cones by construction.
The surface is cut to a single component with disk-topology.
As the set of cones $C$ is admissible, we have $\sum_j \Theta_j = 2\pi(2-2g)$; the hole chain cutgraph has $2g-1$ nodes with 4 corners each, i.e. a total of $8g-4$ corners. Thus Eq.~\eqref{eq:gaussdisk} is satisfied.
\qed

\paragraph{Odd-Couple Condition}
We impose one condition (besides terminals being partners) on the way the numbering of holes is chosen: there needs to be at least one odd couple, i.e. two partner holes which have an odd number of holes in between them in the chain, i.e.\ there is an $i$ and an integer $k$ such that $h_i$ and $h_{i+2k}$ are partners. This will be expected in the proof of equalizability. Note that this is impossible if there are just four or less holes, thus instead special case variations of the hole chain are used for genus 1 (two holes) and genus 2 (four holes) cases, as detailed in Sec.~\ref{sec:special}.

\begin{definition}[Fourfold Cones]
A set of cones $C = \{(c_i,k_i)\}$ with $k_i$ divisible by 4 for each $i$ is called \emph{fourfold}.
\end{definition}

While the above hole chain cutgraph is not equalizable in general, it permits equalization for \emph{specific} righthand sides $\bm{b}$:
\begin{prop}
\label{prop:fourfold}
For any genus, a fourfold cone prescription implies a righthand side $\bm{b}$ for which problem \eqref{eq:equas} is feasible.
\end{prop}
A proof is given in Sec.~\ref{sec:eqproof}. For the general, non-fourfold case, variations of this basic cutgraph are used, as detailed in the following.

\begin{figure}[bt]
\centering
\begin{overpic}[width=0.99\columnwidth]{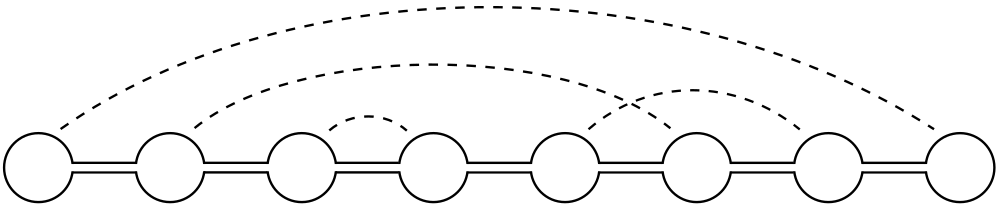}
\end{overpic}
\vspace{-0.2cm}
\caption{Schematic depiction of a chain of holes for a genus $g=4$ surface: circles depict holes (emerging from cutting the surface along $g$ loops), straight line segments depict the sides of cut paths (connectors) between these holes. Together, this hole chain cutgraph cuts the surface to a topological disk, i.e. a sphere with one hole (bounded by the closed black curve).  An exemplary partnership of holes (due to each loop corresponding to two holes) is indicated by dashed arcs; depending on the chosen ordering of holes in the chain, these partner arcs will look different.}
\label{fig:holechain}    
\vspace{-0.2cm}
\end{figure}  

\subsection{General Case (Genus 3+)}
\label{sec:gen3}

If the cone prescription is not fourfold, i.e. there is at least one $k_i\!\!\mod 4 \neq 0$, we extend the hole chain cutgraph by one extra path (cf.~Fig.~\ref{fig:extracutB}) -- which makes it equalizable.

\begin{definition}[Valid Extra Path]
A simple path is called a valid extra path for a hole chain cutgraph if
\begin{itemize}
\item it does not cross any cone,
\item only its endpoints intersect the hole chain cutgraph,
\item at least one endpoint is on a hole of the hole chain,
\item no endpoint is coincident with a node of the hole chain.
\end{itemize}
\end{definition}
Notice that this extra path forms two additional nodes, both of degree 3, i.e. T-nodes, at its endpoints -- at each we mark as corners the two sectors directly adjacent to the extra path. Hence the total number of corners increases from $8g-4$ to $8g$. At the same time, the extended hole chain cutgraph cuts the surface into two components, $M_0'$ and $M_1'$, each with disk-topology, with numbers of corners $n_0, n_1$. We have $n_0+n_1 = 8g$.

In order for the extended hole chain to remain admissible, we need to ensure that Eq.~\eqref{eq:gaussdisk} is satisfied, i.e. the number $n_i$ of corners per component $M_i'$ needs to match the total curvature of cones $C_i'$ prescribed within the component. Note that this is satisfied for $M_0'$ if and only if  it is satisfied for $M_1'$. Also, if the endpoints lie on two mated segments, we choose them as mated points so as to create a degree 4 node rather than two opposite T-nodes. We furthermore require that the numbers $n_0$, $n_1$ of corners are not divisible by 4. This will be expected in the proof of equalizability.

\begin{definition}[Admissible Extra Path]
A valid extra path that yields corner numbers $n_0$, $n_1$ not divisible by 4 and satisfying Eq.~\eqref{eq:gaussdisk} is called admissible.
\end{definition}

\begin{figure}[b]
\centering
\begin{overpic}[width=0.99\columnwidth]{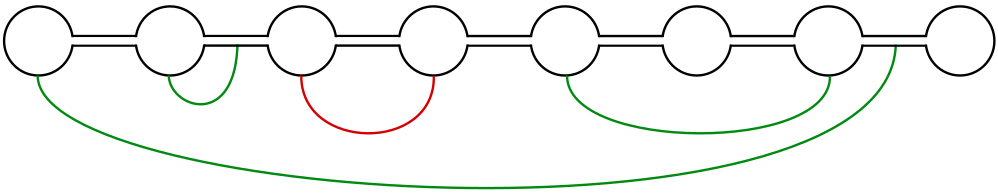}
\end{overpic}
\vspace{-0.05cm}
\caption{Examples of extra cuts (green) that could be added to the hole chain cutgraph. The red path is not a valid extra cut because it splits the surface into two components with $n_0 = 4$ and $n_1 = 8g-n_0 = 28$ corners (cf.~Sec.~\ref{sec:gen3}).}
\label{fig:extracutB}    
\end{figure}  

\begin{prop}
\label{th:extraexists}
An admissible extra path exists for any non-fourfold cone prescription and any genus $g \geq 3$.
\end{prop}
\proof
Pick one prescribed cone $c_i$ with $k_i\!\!\mod 4 \neq 0$. Let $\beta$ be a simple path from a point $q$ on hole $h_0$ which is not a corner to $c_i$ such that it does not contain any other cone. Let $\gamma$ be a path that starts at $q$, runs (arbitrarily close) along one side of $\beta$, then around $c_i$, then back along the other side of $\beta$, and ultimately (arbitrarily close) along the cut $G$ until it has passed $k_i-2$ corners. It then connects to a point $q'$ on the segment it reached. If $\gamma$ is chosen sufficiently close to $\beta$ and $G$, it is an example of a valid extra path: the region that contains $c_i$ contains no other cone and it has $k_i$ corners (the $k_i-2$ corners passed along $G$ plus the two corners formed by the extra path itself with $G$ at $q$ and $q'$). Also, $\gamma$ is connected to a hole, namely $h_0$.\qed

The equalizability of the hole chain cutgraph extended by an admissible extra path is proven in Sec.~\ref{sec:eqproof}.

\subsection{Special Cases (Genus 0, 1, 2)}
\label{sec:special}

\subsubsection*{Genus 0 Case}

In the case of a topological sphere, our method formally is applicable, but does not actually contribute anything: the cutgraph is empty; there are no cuts across which the cone metric could be non-seamless, thus no padding is required. The existence of conformal metrics with prescribed cones on the topological sphere $M$ is well-known \cite{Troyanov:1991}.}

\begin{figure}[t]
\centering
\begin{overpic}[width=0.99\columnwidth]{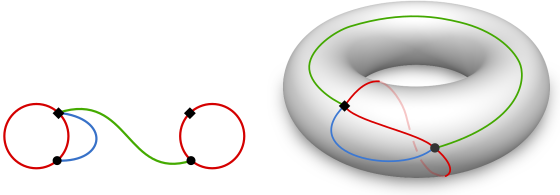}
\end{overpic}
\vspace{-0.35cm}
\caption{Cutgraph pattern for genus 1 surfaces, shown abstracly (left) and on an example surface (right). The surface is partitioned into a 2-corner region (enclosed by blue and red paths) and a 6-corner region.}
\label{fig:gen1}    
\end{figure}  

\subsubsection*{Genus 1 Case}
\label{sec:gen1}

For the case of genus 1 surfaces, we (like in the general case) add one extra path, but deviate slightly from the general hole chain pattern in terms of identification of connector endpoints. The cutgraph pattern is depicted in Fig.~\ref{fig:gen1}.
Notice that the surface is split into a component with 2 corners and a component with 6 corners. It is easy to see that for any admissible prescription of cones on a genus 1 surface, one either has no cones at all (in this case the basic hole chain cutgraph is sufficient, cf.~Prop.~\ref{prop:fourfold}) or one has, among the prescribed cones, one or more cones whose curvature sums up to $\pi$ (due to Gauss-Bonnet there are cones of positive and of negative curvature, and the case of a single positive cone of curvature $\nicefrac{\pi}{2}$ is the one non-admissible case, cf.~Sec.~\ref{sec:intro}). The surface bi-partition by this cutgraph pattern is thus compatible with any non-empty cone prescription, i.e. the paths can be chosen in an admissible way on $M$.

The equalization equation system corresponding to this pattern is easily checked explicitly for non-negative feasibility (cf.~Sec.~\ref{app:g1}).

\begin{figure}[b]
\centering
\begin{overpic}[width=0.95\columnwidth]{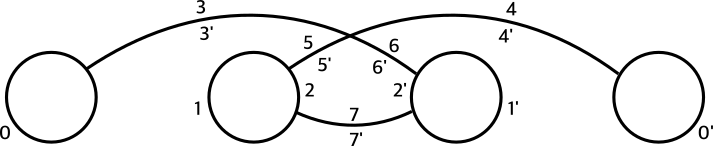}
\end{overpic}
\vspace{-0.05cm}
\caption{One of the cutgraph patterns for genus 2 surfaces. Segments $i$ and $i'$ are mates, i.e. correspond to a common cutgraph branch. The surface is partitioned into a 5-corner region (center) and a 11-corner region (surround).}
\label{fig:gen2val5}    
\end{figure}  

\subsubsection*{Genus 2 Case}
\label{sec:gen2}

For the case of genus 2 surfaces, we (like in the genus 1 case) need to deviate from the general case. In contrast to the genus 1 case, where we could assume that a subset of prescribed cones always have curvatures summing to one specific value, five cases need to be distinguished: there is a subset of prescribed cones with curvatures summing to $\pi$, $\nicefrac{\pi}{2}$, $-\nicefrac{\pi}{2}$, $-\pi$, or $-3\nicefrac{\pi}{2}$ (compatible with regions with 2, 3, 5, 6, and 7 corners, respectively).
This list is exhaustive because the total sum of cone curvatures is $-4\pi$ on a genus 2 surface, and not all cones have curvatures that are multiples of $2\pi$ (as this case was handled already in Prop.~\ref{prop:fourfold}); thus at least one of these five values has to appear as a subsum.

Depending on which curvature sum subset is available in a given set of prescribed cones, the cutgraph pattern needs to be chosen compatibly. For the case that a cone subset with curvature sum $-\nicefrac{\pi}{2}$ is available (as in most practical scenarios), the pattern depicted in Fig.~\ref{fig:gen2val5} can be used; for the remaining patterns refer to Sec.~\ref{app:g2}. Notice that this pattern is a variation of the basic hole chain: two connectors are required to cross. This partitions the surface into a 5-corner region (compatible with a curvature $-\nicefrac{\pi}{2}$ subset) and a 11-corner region (compatible with the remaining cones).

The equation systems corresponding to these patterns are easily checked explicitly for non-negative feasibility (cf.~Sec.~\ref{app:g2}).

\section{Implementation}
\label{sec:implement}

We now describe how our algorithm can be implemented for discrete surfaces. Thus, in the following, $M$ represents a closed triangle mesh of arbitrary genus $g$. Cones are prescribed on vertices of $M$; such vertices are called \emph{cone vertices}.

\paragraph{General Overview} The overall algorithmic steps are:
\begin{enumerate}
\item Construct $g$ non-contractible loops and cut $M$ along these loops (Sec.~\ref{sec:61}).
\vspace{0.15cm}
\item Connect all holes using shortest paths, selecting the connection pattern based on the given genus and cone prescription.
If necessary, add one extra path;  then cut the mesh (Sec.~\ref{sec:61}).
\vspace{0.15cm}
\item Set target angles at cone and corner vertices. Compute a corresponding discrete conformal metric (Sec.~\ref{sec:62}).
\vspace{0.15cm}
\item Number all cut segments and set up the system matrix $A$ and righthand side $\bm{b}$ accordingly (Sec.~\ref{sec:63}).
\vspace{0.15cm}
\item Compute a solution to the linear system $A\bm{w} = \bm{b}$; add a constant shift to yield a solution $\bm{w} \geq 0$ (Sec.~\ref{sec:63}).
\vspace{0.15cm}
\item Extend the cutgraph to include all cones; lay out the mesh in the plane according to the metric (Sec.~\ref{sec:64}).
\vspace{0.15cm}
\item Perform padding according to the computed padding widths~$\bm{w}$ (Sec.~\ref{sec:65}).
\vspace{0.15cm}

\end{enumerate}

In the end, this yields seamless parametrizations with prescribed cones. While the parametric distortion may initially be high,
these parametrizations provide the \emph{feasible starting point} required by techniques for injectivity-preserving parametrization distortion optimization (cf.~Sec.~\ref{sec:66})
as well as  quadrangulation methods based on quantization of seamless parametrizations.

\subsection{Cutgraph}
\label{sec:61}
On $M$, one option to obtain $g$ non-intersecting non-contractible loops is via handle or tunnel loop algorithms \cite{DeyFW13}, modified to avoid cone vertices. A simpler robust approach is to iteratively cut the mesh $g$ times, each time by an arbitrary non-contractible loop, not containing a cone vertex or a boundary vertex, obtained using the tree-cotree algorithm \cite{Erickson05greedyoptimal}. For the sake of ultimately yielding a cutgraph that is not unnecessarily long and convoluted, it is advisable to pick a short loop each time.

For the same reason we construct the connectors as shortest paths, not containing cone, boundary, or other paths' vertices, between the $2g$ holes using Dijkstra's algorithm. A natural ordering of the holes in the chain can be determined using a Hamiltonian path algorithm; this order then needs to be adjusted slightly before the connectors are constructed, to ensure the paired-terminals condition and the odd-couple condition are satisfied (cf.~Sec.~\ref{sec:holechain}).

\paragraph{Mesh Refinement}
Here and in the following we work with discrete paths/loops, following the edges of the mesh $M$. To be robust regardless of the mesh structure, after each construction of a path we split each mesh edge that is not on a path if its two vertices both are either on the boundary, on a path, or on a cone, thereby ensuring that the mesh, cut by the paths, remains path-connected even in the sense of the discrete edge paths (avoiding boundary, path, and cone vertices) we use here. 

\paragraph{Extra Cut}
The extra cut path (needed for the general genus $g \geq 3$ case) is likewise constructed as a shortest path. However, we need to employ a cone-aware variant of Dijkstra's algorithm in order to ensure cone/corner compatibility, cf.~Eq.~\eqref{eq:gaussdisk}. 

To this end, to each directed dual edge $e$ of the mesh $M'$ we assign a value $\rho_e$ (with $\rho_e = -\rho_{\bar e}$ for oppositely directed dual edges $e$, $\bar e$) such that the sum of these values clockwise around a single cone vertex $c_i$ is $\hat\Theta_i$, and around non-cone vertices zero. Such an assignment can, for instance, be achieved using a spanning tree of the cones (dashed in Fig.~\ref{fig:periodtree}), rooted at the boundary $\partial M'$: initialize $\rho$ at all leaves of the tree and propagate the values towards the root, summing values where branches meet.

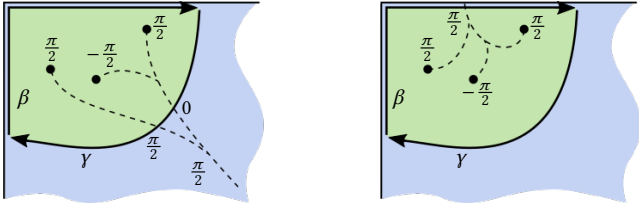
\begin{figure}[tb]
\centering
\begin{overpic}[width=0.99\columnwidth]{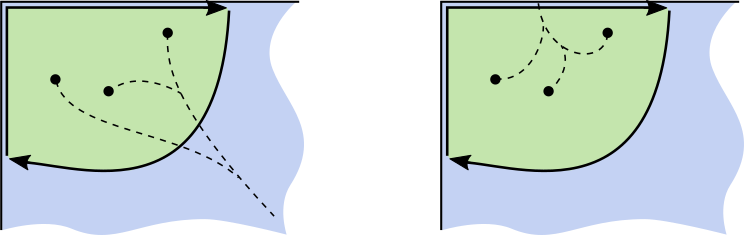}
\put(12,6.2){\small$\gamma$}
\put(2,15.5){\small$\beta$}
\put(6,23.5){\small$\frac{\pi}{2}$}
\put(13,22.5){\small$-\frac{\pi}{2}$}
\put(23,26.9){\small$\frac{\pi}{2}$}
\put(28,14){\small$0$}
\put(22,8){\small$\frac{\pi}{2}$}
\put(29,4){\small$\frac{\pi}{2}$}

\put(71,6.2){\small$\gamma$}
\put(61,15.5){\small$\beta$}
\put(65,23.5){\small$\frac{\pi}{2}$}
\put(72,16.5){\small$-\frac{\pi}{2}$}
\put(82.3,26.9){\small$\frac{\pi}{2}$}
\put(69.2,27.6){\small$\frac{\pi}{2}$}

\end{overpic}
\vspace{-0.05cm}
\caption{Example of holonomy-aware extra path computation. Left: a tree of cones with computed $\rho$-values is shown in black, dashed. Path $\gamma$ from boundary to boundary, crossing two tree branches, has a holonomy value $\sum_\gamma \rho = \nicefrac{\pi}{2}$. This path is closed along the boundary by $\beta$ (with $\sum_\beta\rho = 0$), forming $n=3$ corners. As $\sum_{\gamma+\beta}\rho = \nicefrac{\pi}{2}$ and $n=3$ conforms with Gauss-Bonnet \eqref{eq:gaussdisk}, the path $\gamma$ can be accepted. Right: to illustrate that the tree of cones can be chosen arbitrarily, here the same situation is depicted with a different tree. We have $\sum_{\gamma}\rho = 0$ and $\sum_\beta\rho = \nicefrac{\pi}{2}$, thus again $\sum_{\gamma+\beta}\rho = \nicefrac{\pi}{2}$.}
\label{fig:periodtree}    
\end{figure}  

Now for an arbitrary simple closed clockwise dual edge path $\gamma$, we have the following important property: $\sum_{e\in\gamma} \rho_e = \sum_{p_i\in\Gamma} \hat\Theta_i$, where $\Gamma$ is the set of all vertices enclosed by $\gamma$ \cite{CraneDS10}. This sum $\sum_{e\in\gamma} \rho_e$ is called, with slight abuse of terminology also for non-closed paths, \emph{holonomy} in the following. Intuitively, the sum $\sum_{e\in\gamma} \rho_e$ along a closed path tells us what total cone curvature is contained in the region enclosed by the path.

We then employ Dijkstra's algorithm, starting from a hole segment on $\partial M'$, and keep track of the holonomy values along the way. Whenever Dijkstra's front propagation reaches $\partial M'$ again, we tentatively close the loop by walking back to the starting point clockwise along $\partial M'$, counting passed corners on the way, and checking whether the total holonomy matches the number of corners, cf.\ Sec.~\ref{sec:cg}, Eq.~\eqref{eq:gaussdisk}. If it matches, the path is accepted and added as extra cut. An example is shown in Fig.~\ref{fig:periodtree}.

The following needs to be taken into account in this though: Dijkstra's algorithm keeps track of, for each vertex, the shortest path back to the starting point -- regardless of holonomy $\sum \rho$.
So while there are shortest paths of different holonomy back to the starting point, Dijkstra's algorithm discards all but the shortest one.
 We, instead, need to keep track of the shortest path \emph{per vertex per holonomy}. Otherwise paths that could end up having a suitable holonomy in the end, could already be discarded in favor of a shorter path with an unsuitable holonomy. We thus perform Dijkstra's algorithm not on $M'$, but on a branched covering of $M'$, with sheets glued according to $\rho$ \cite{kalberer2007qsp}. In practice this simply means that each triangle stores separate distance information 
 per holonomy class, indexed by the 
 value $\sum \rho$
 of incoming fronts.
 
 An additional measure needs to be taken because we need $\gamma$ to be simple on $M'$. While the holonomy-aware version of Dijkstra's algorithm yields a path that is simple on the covering, its projection down onto $M'$ may be self-intersecting. Before advancing the front to the next vertex we thus always check whether this vertex is already contained in the predecessor path to prevent such self-intersections.
 While with this modification one can no longer guarantee that a path is always found, one can always fall back to an explicit path construction following the existence proof in Sec.~\ref{sec:gen3}; we have never encountered a case where this was necessary.

\paragraph{Special Cases}

The connectors of the special cutgraph patterns employed for genus 1 and genus 2 surfaces are realized using shortest paths as well. These are constructed incrementally between endpoints chosen on the holes, and cross points chosen on other connectors. For those paths that finally split the surface, again the above cone-aware shortest path algorithm is employed to ensure cone/corner compatibility.

\subsection{Conformal Map}
\label{sec:discconf}
\label{sec:62}

After cutting $M$ using the cutgraph $G$ to obtain $M'$, for each component of $M'$ (typically one or two, except for some genus 2 cone configurations) we need to obtain a cone metric with rectilinear boundary.

This is the one part of the implementation, where --  even though in the continuous case (cf.~Sec.~\ref{sec:conemetrix}) things are rather straightforward -- achieving robustness is actually a challenge.
In contrast to the continuous case, questions of existence (or even just precise definition) of \emph{discrete} conformal metrics with prescribed cones and boundary curvature are not entirely settled.
For cases without boundary, recent results have brought insight \cite{Luo:2004,Gu:2013,Gu:2014,Springborn:2017}, but the case with boundary requires further work on the theory side.

A method for handling the discrete case without boundary was recently described in \cite{Campen:2017:SimilarityMaps}; it combines the elegant method proposed in \cite{Springborn:2008} with on-demand mesh modifications (edge flips, following \cite{Luo:2004}). It can be used for the genus 0 case as is.
With a minor extension (as follows) we additionally prescribe geodesic boundary curvature, and find that it works for this case as well.
But as mentioned earlier, further work is necessary to determine whether guarantees (regarding general existence, termination, etc.) can be formally established. We remark that our overall seamless parametrization construction does not at all rely on the cone metric being conformal -- this was merely a convenient natural choice for the theoretical considerations -- thus the concrete notion of discrete conformal equivalence employed is irrelevant in our context, which may simplify the situation.

The boundary curvature prescription is performed using the holonomy angle constraints offered by this method; however, not for homology loops, but for each boundary vertex's triangle fan, prescribing $\nicefrac{\pi}{2}$ at corners and $\pi$ at all other boundary vertices.

Another recently proposed conformal mapping method \cite{Sawhney:2017:BFF} supporting cone and boundary curvature prescription is particularly efficient -- but does not include remeshing capabilities inevitably required for full robustness.  One could think about employing a hybrid solution in practice, where the efficient algorithm is tried first while the more general one serves as fallback.

Note that, if this is important in a use case, the edge flips which are performed by the conformal metric computation algorithm can ultimately be realized by means of edge splits, as described in \cite{Fisher:2007}. In this way the output mesh is a locally refined version of the input mesh (rather than a mesh with arbitrarily different combinatorial structure) and its embedding is easily preserved.

\subsection{Equalization}
\label{sec:63}
We (arbitrarily) number the segments of $\partial M'$ and set up the system matrix $A$ \eqref{eq:equas} accordingly, with one equation $\eqref{eq:equa}$ for each pair of mates (or equation $\eqref{eq:equaGeneralw}$ where T-nodes are involved, cf.~Sec.~\ref{sec:gen3}).
The righthand side $\bm{b}$ is determined by measuring the lengths of the segments under the metric computed in Sec.~\ref{sec:discconf}.

Then the linear system $A\bm{w}=\bm{b}$ is solved. As it is underdetermined, we compute the least-norm solution $\bm{w}^*$ via $A^TA\bm{w}=A^T\bm{b}$. 
The resulting solution $\bm{w}^*$ does not generally satisfy the important non-negativity constraint of problem \eqref{eq:equas}.
However:
\begin{prop}
\label{th:constant}
If cutgraph $G$ contains no T-nodes, $\bm{w}^+ = \bm{w}^* + \lambda\bm{1}$, with $\lambda = -\min \bm{w}^*$, satisfies \eqref{eq:equas}, i.e. $A\bm{w}^+=\bm{b}$ and $\bm{w}^+ \geq 0$.
\end{prop}
\proof{Each equation $\eqref{eq:equa}$, thus each row of $A$, involves two (not necessarily distinct) variables with positive sign and two with negative sign. We thus have $A\bm{1} = \bm{0}$, i.e. globally constant padding (e.g. by width 1) is in the kernel of $A$, it does not affect equalization. Hence, $A\bm{w}^+ = A\bm{w}^* + \lambda A\bm{1} = A\bm{w}^* = b$ for any $\lambda$. With the above choice of $\lambda$, obviously $\bm{w}^+_i \geq 0\,\forall i$.}\qed

This means we can add a sufficiently large global value to the initial solution $\bm{w}^*$ in order to obtain a non-negative solution.

One class of cutgraphs (cf.~Sec.~\ref{sec:gen3}) contains one extra path that forms two T-nodes. This leads to one or two general equations \eqref{eq:equaGeneralw} with four or six negative variables.
Note that the additional negative sign entries correspond to the variables $w_x, w_{x'}$ of the extra path's segments $s_x, s_{x'}$ -- if we removed these variables, each equation would have two positive and two negative variables. Hence, while $A\bm{1} \neq \bm{0}$, we have $A\bm{1'} = \bm{0}$, where $\bm{1'}$ is the vector of ones except for zeroes at positions $x$ and $x'$.
Let $\bar A$ be the matrix obtained from $A$ by zeroing the two columns corresponding to $w_x, w_{x'}$.

\begin{prop}
\label{th:constant}
If cutgraph $G$ is a hole chain with extra path, $\bm{w}^+ = \bm{w}^* + \lambda\bm{1'}$, with $\lambda = -\min \bm{w}^*$ and $\bm{w}^*$ the least-norm solution of $\bar A \bm{w} = \bm{b}$, satisfies \eqref{eq:equas}, i.e. $A\bm{w}^+=\bm{b}$ and $\bm{w}^+ \geq 0$.
\end{prop}
\proof{
$A \bm{w} = \bm{b}$ is feasible with $w_x = w_{x'} = 0$ as proven in Sec.~\ref{sec:eqproof}, thus $\bar A \bm{w} = \bm{b}$ is feasible and $\bm{w}^*$, being a least-norm solution, has $w^*_x = w^*_{x'} = 0$, thus $\bar A \bm{w}^* = A \bm{w}^*$. 
Hence, $A\bm{w}^+ = A\bm{w}^* + \lambda A\bm{1'} = \bar A\bm{w}^* = \bm{b}$ for any $\lambda$. With the above choice of $\lambda$, obviously $\bm{w}^+_i \geq~0\,\forall i\neq x,x'$, and $w^+_x = w^+_{x'} = 0$.}\qed

\subsection{Flattening}
\label{sec:64}

In order to obtain the map $F$ according to the conformal cone metric, we first need to extend the cutgraph $G$ to $G^T$ (cf.~Sec.~\ref{sec:mtop}).
In each component of $M'$ we pick a point $p$ (which is not a corner) on some segment, and compute shortest paths from $p$ to all cones in the component. The union of these paths forms the tree $T$ to extend $G$ to $G^T$. Note that the  piecewise-linear form of padding we use in the discrete setting (cf.~Sec.~\ref{sec:65}) does not require $T$ to meet the segment specifically at right angles.

The conformal metric computed in Sec.~\ref{sec:discconf} is then flat in all of $M^c$, so its components can be laid out in the plane \cite{Springborn:2008}, isometrically with respect to the metric, to obtain $F$.

\begin{figure}[tb]
\centering
\begin{overpic}[width=0.99\columnwidth]{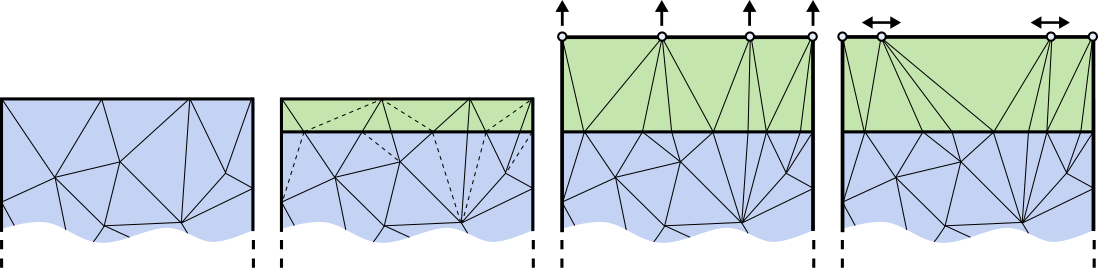}
\put(10.5,-2){\small a)}
\put(36,-2){\small b)}
\put(61.5,-2){\small c)}
\put(87,-2){\small d)}
\end{overpic}
\caption{a) Mesh near a segment (top) that is to be padded. b) The strip to be stretched (green) is formed by inserting a straight line into the triangulation (by splitting edges at the intersections), so close to the segment that no vertex is contained. c) The strip is stretched outwards by displacing the vertices that lie on the segment by the desired padding width. d) The vertices on the segment are translated laterally according to $\phi$ (cf.~Sec.~\ref{sec:padding}) to achieve pointwise seamlessness.
}
\label{fig:paddiscrete}    
\end{figure}

\begin{figure*}[t]
\centering
\begin{overpic}[width=0.99\linewidth]{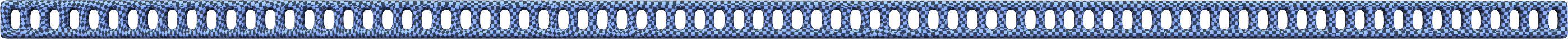}
\end{overpic}
\vspace{-0.25cm}
\caption{A locally injective seamless map generated on an 80-torus.}
\label{fig:80torus} 
\vspace{-0.2cm}   
\end{figure*}  

\begin{figure}[b]
\vspace{-0.1cm}
\centering
\begin{overpic}[width=0.99\linewidth]{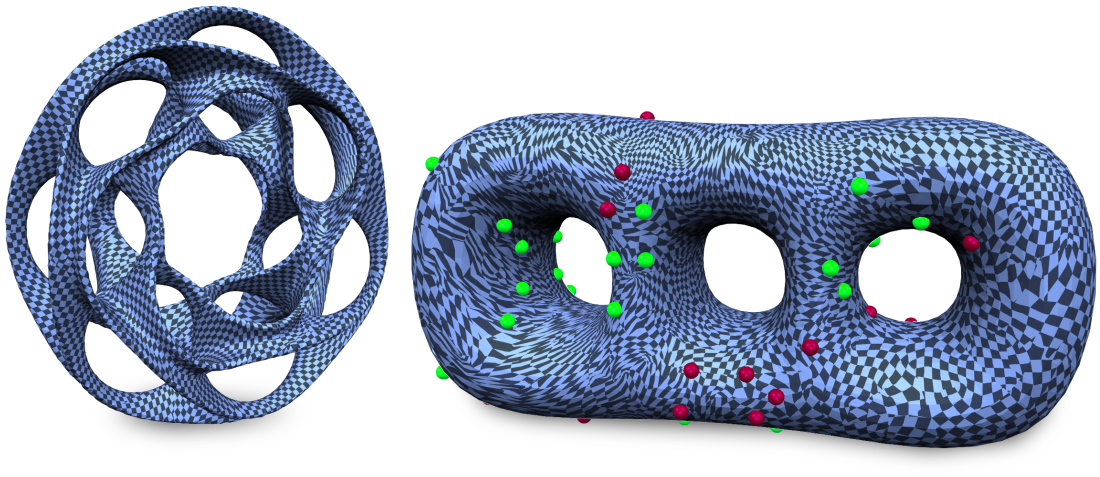}
\end{overpic}
\vspace{-0.2cm}
\caption{Left: example map generated on a topologically complex surface. Right: Example map generated with geometrically non-meaningful cone prescription (here: 50 randomly distributed cones of curvatures $\pi$ and $-\pi$).}
\label{fig:random}    
\end{figure}  

\begin{figure*}[t]
\centering
\begin{overpic}[width=0.91\linewidth]{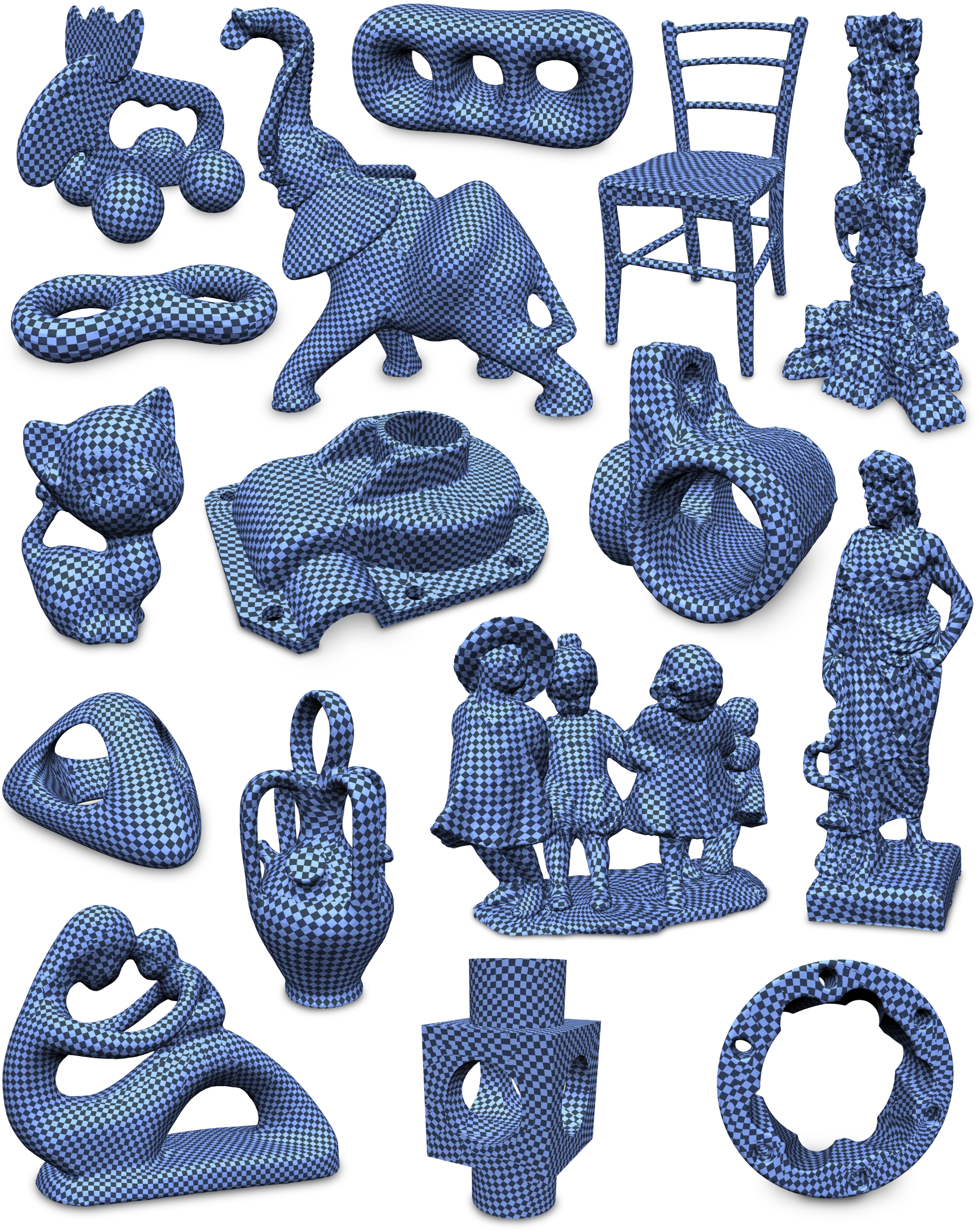}
\end{overpic}
\caption{Visualization of a variety of locally injective seamless parametrizations obtained using our method}
\label{fig:results}    
\vspace{-0.5cm}
\end{figure*}

\subsection{Padding}
\label{sec:65}

The padding operation described in Sec.~\ref{sec:pad1} can easily be performed in a piecewise linear setting as illustrated in Fig.~\ref{fig:paddiscrete}: After inserting a parametrically straight line into the mesh $M$ along a segment (such that no vertex is contained in the resulting strip), the stretching and lateral shifting can be performed by simply relocating those vertices that are on the segment. Note that, as no vertices lie in the strip, the mesh within the rectangular strip is a simple triangle strip and laterally translating the segment vertices does not cause triangle inversions as long as their order is preserved -- which is the case with the reparametrization $\phi$ (cf.~Sec.~\ref{sec:padding}).

We note that the resulting seamless parametrization is not of immediate practical use: the scale distortion involved in the conformal map together with the additional padding-induced stretching often leads to high parametric distortion.
This map, however, provides the valid (locally injective and seamless) starting point required by robust optimization methods that can convert it to a low-distortion parametrization (to the extend permitted by the cone prescription) as detailed in the following.

\subsection{Distortion Optimization}
\label{sec:66}

For the optimization of the seamless padded map, in our implementation we use the symmetric Dirichlet energy together with quadratic proxies for efficiency as described in \cite{Rabinovich:2017:SLI}. This method preserves local injectivity during optimization by design, and we additionally impose linear seamlessness constraints to preserve seamlessness of the map:
$$\vec{e_i} = R^{k_{ij}\frac{\pi}{2}}\vec{e_j}\quad\text{for each pair $(i,j)$ of identified mesh edges},$$
where $\vec{e_i}$ is the edge vector of edge $i$ in the parametric domain, and $R^{k_{ij}\frac{\pi}{2}}$ is a rotation by $k_{ij}\frac{\pi}{2}$, where the constant integer $k_{ij}$ is determined by the edges' relative initial orientation in the  domain.

In essence, our method provides the feasible starting point required by such nonconvex distortion minimization techniques.

\section{Examples}

We demonstrate our implementation of the presented algorithm on a number of examples.

Figure \ref{fig:results} shows a visualization of the seamless parametrizations constructed on models from the dataset provided by Myles et al.~\shortcite{Myles:2014}. We employed the cone position and angle prescriptions provided with that dataset.
Figure \ref{fig:random} demonstrates the algorithm handling topologically complex surfaces as well as randomly prescribed singularities.
In all cases locally injective seamless global parametrization were obtained. Note that seamless does not mean that cuts are not visible in these checkerboard visualizations; for this the maps would additionally need to be quantized \cite{CampenBK15} -- a process for which our method provides a suitable initialization.

In order to explore the numeric limits of our implementation, we applied it to $N$-tori, for increasing $N$. For an 80-torus, as depicted in Figure \ref{fig:80torus}, the implementation succeeds; for a 100-torus we are still able to obtain an initial seamless map -- however, with a level of distortion that state-of-the-art local injectivity preserving optimization methods prove to have trouble with, due to numerical precision issues. For even larger $N$, the computation of the constrained conformal map starts to occasionally suffer from numerical issues (e.g. step size going down to numerically zero) as well. Investigation of such numerical aspects of map optimization in high distortion cases is  an important direction for future research.

\section{Conclusion and Future Work}
\label{sec:conclusion}

This paper provides a general path to obtaining seamless parametrizations with a given set of cones. On a conceptual level the approach is simple (just pad a cut-aligned map), and we hope it provides some new insights into how global parametrizations can be constructed. 

A conceptual limitation in its current form is its unawareness of holonomy angles on homology loops (in addition to local cone angles), which, for instance, is important for parametrizations following a global guiding field. We expect that by using different forms of cutgraph construction, based on given global holonomy angles, many of the ideas herein will be applicable to such a setting as well; we plan to address this is a separate paper. 

Related directions of future work include generalization to surfaces with boundaries as well as aligning to tagged feature curves or other prescribed directions on the surface.

In the smooth setting, our algorithm is guaranteed to always yield a valid, locally injective seamless parametrization. In practice, numerical optimization routines notoriously bring about challenges due to numerical precision limits, which applies here to the discrete conformal map computation.
As in this discrete setting there furthermore are some unclarities concerning definition and existence of general conformal maps, a potential path could be the replacement of this initial map computation with a different technique -- exploiting the fact that conformality is not actually required.

\vspace{-0.05cm}
\begin{acks}
H. Shen and D. Zorin are supported by awards NSF IIS-1320635, NSF DMS-1436591, and a gift from Adobe; J. Zhou by a CSC scholarship.
\end{acks}

% Bibliography
\vspace{-0.05cm}
\bibliographystyle{ACM-Reference-Format}
\bibliography{param2,param}

\appendix
\section{Proofs of Equalizability}
\label{sec:eqproof}

\subsection{Genus 3+}

The following proof is constructive, yet is only intended to prove feasibility. In practice a simple linear system solve can be used to obtain a solution instead (cf.~Sec.~\ref{sec:63}).

\subsubsection*{Non-Fourfold Case}

Cutgraph $G$ together with the extra path cuts $M$ into two components $M_0'$ and $M_1'$, neither of which has its number 
of segments $m_k$ divisible by 4. We 
assume that a segment of one hole which is in between the odd-couple pair (cf.~Sec.~\ref{sec:holechain}) is in $M_1'$.

Let the segments of each component $M_k'$ be numbered counterclockwise along its boundary, in $M_0'$ from $0$ to $m_0-1$, in $M_1'$ from $m_0$ to $m_0+m_1-1$.
Let $B\bm{w} = \bm{c}$ be a linear system of equations
\begin{figure}[h!]%{r}{0.14\linewidth}
  \vspace{-0.14cm}
  \hspace{-0.45cm}
         \hfill\begin{overpic}[width=0.168\linewidth]{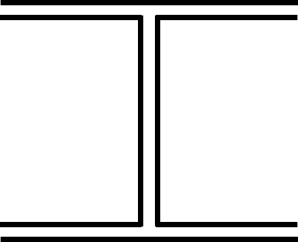}%1.2\linewidth]{T2}
    \put(35,36){\small 1}
    \put(57,36){\small 1}
    \put(20,59){\small 1}
    \put(72,59){\small 0}
     \put(46,85){\small 1}
          \put(2,36){\small $M_0'$}
          \put(75,36){\small $M_1'$}
    \put(20,12){\small 1}
    \put(72,12){\small 0}
 \put(46,-15){\small 1}
    \end{overpic}
  \vspace{-0.2cm}
\end{figure}
\vspace{-1.55cm}
\begin{minipage}[l]{0.78\linewidth}
$w_{\text{prev}(i)} + w_{\text{next}(i)} = d_i - \ell_i$, $i=0,\dots,m_0+m_1-1$, 
where $d_i = 1$ for each segment, except for the two segments in $M_1'$ which are adjacent to (but not on) the extra path: for these $d_i = 0$, as depicted here.
\end{minipage}

\vspace{0.1cm}
Notice that the matrix $B$ has block structure, as each equation concerns one of the two components only and segment indices are ordered by component:
\vspace{-0.025cm}$$B=\left[\begin{array}{cc}B_0&0\\0&B_1\end{array}\right].$$

\begin{lemma}
If $\bm{w}$ satisfies $B \bm{w} = \bm{c}$ then it satisfies $A\bm{w} = \bm{b}$ of \eqref{eq:equas}.
\end{lemma}
\proof
System $B \bm{w} = \bm{c}$ enforces unit length for all segments (except zero length for one side of each T-node). As this is a (special) case of length equalization, equations \eqref{eq:equa} and \eqref{eq:equaGeneralw}, which $A\bm{w} = \bm{b}$ consists of, are satisfied by $\bm{w}$ that solves $B \bm{w} = \bm{c}$ .\qed

\begin{lemma}
System matrix $B$ has full rank.
\end{lemma}
\proof
It suffices to show that both blocks, $B_0$ and $B_1$, have full rank.
Each $B_k$ is a \emph{circulant} matrix, with \emph{associated polynomial} $f(x) = x + x^{m_k-1} = x(x^{m_k-2}+1)$ \cite{davis2012circulant}. It is full rank whenever 
$f(e^{\nicefrac{\,j}{m_k}2\pi\mathrm{i}}) \neq 0$, for $j=0\ldots m_k-1,$ 
as its determinant is given by the product of these values ($f$ on $m_k^\text{th}$ roots of unity). It is straightforward to check that $f(e^{\nicefrac{\,j}{m_k}2\pi\mathrm{i}}) = 0$ (for some $j$) requires $e^{\nicefrac{\,j(m_k-2)}{m_k}2\pi\mathrm{i}} = -1$, which requires $2\nicefrac{j(m_k-2)}{m_k}$ to be odd, which requires $\nicefrac{4j}{m_k}$ to be odd, which requires $m_k$ to be a multiple of 4, which is not the case here by construction of the extra cut path.
\qed

Hence, $B$ admits a (unique) solution $\bm{\hat{w}}$.
Notice, however, that $\bm{\hat{w}}$ is not generally non-negative, and -- as the padding variables of the two extra path segments are not zero (and potentially negative) -- we cannot simply remedy this by a constant shift as in Prop.~\ref{th:constant}.

We thus proceed as follows:
\begin{itemize}
\item By choosing different (non-unit) $d_i$ for righthand side $\bm{\bar c}$ (which, however, still imply length equalization), we can cause one of the two extra cut path padding variables to be zero.
\item By violating one equation of the system, we can bring the second one to zero as well.
\item Exploiting the odd-couple condition, we then modify the padding widths to repair the violation, in the sense that, while the (unnecessarily strong) conditions of system~$B$ are no longer satisfied, the conditions of system $A$ are.
\item Finally we obtain a non-negative solution to $A\bm{w}=\bm{b}$ through a (now possible) constant shift.
\end{itemize}

Let $e_0$ be the extra cut path segment in $M_0'$. Choose $\mu = - {\hat{w}}_{e_0}$ (the current padding width) and let $\bm{\bar w}_0 = \bm{\hat w}_0 + \mu\bm{1}$; then $(\bm{\bar w}_0)_{e_0} = 0$, i.e. now the padding width of this segment is zero.

As $B\bm{1} = \bm{2}$, we have $B_0 \bm{\bar w}_0  = \bm{\bar c}_0 =  \bm{c}_0 + \mu\bm{2}$.
Solve $B_1 \bm{\bar w}_1 = \bm{\bar c}_1$, where $\bm{\bar c}_1 = \bm{c}_1 \!+\! \mu\bm{2}$ except for two entries with $d_i \!=\! 0$; these remain as in $\bm{c}_1$.

Then it holds:\vspace{-0.2cm}
$$\left[\begin{array}{cc}B_0&0\\0&B_1\end{array}\right]    \left[\begin{array}{c}\bm{\bar{w}}_0\\\bm{\bar w}_1 \end{array}\right]    =     \left[\begin{array}{c}\bm{\bar c}_0\\\bm{\bar c}_1\\\end{array}\right]$$

As $\bm{\bar w}$ leads to each segment being padded to length $1 + 2\mu$ (except for one segment adjacent to each T-stem, which attains zero length), thus again to a (special) case of equalization, we have $A \bm{\bar w} = \bm{b}$.

Let $\bm{1}^{(i)}$ be the vector of all zeroes, except for a single $1$ at position~$i$.
It holds $B_1 \bm{1}^{(i)} = \bm{1}^{(i-1)} + \bm{1}^{(i+1)}$.

Let $\bm{k} = \bm{1}^{(k+1)} - \bm{1}^{(k+3)} + \bm{1}^{(k+5)} - \bm{1}^{(k+7)} \dots + \bm{1}^{(k-1)}$ (where indices are cyclical, i.e. taken $\!\!\!\mod m_1$).
Then $B_1 \bm{k} = 2\cdot\bm{1}^{(k)}$.

Pick any hole segment $k$ in $M_1'$. Let $\bm{\tilde w}_1 = \bm{\bar w}_1 -\frac{\bm{\bar w}_{e_1}}{\bm{k}_{e_1}} \bm{k}.$
Then $(\bm{\tilde w}_1)_{e_1} = 0$, i.e. now the padding variable of the extra cut path's other segment is zero, too.

Note that picking a hole segment rather than a connector segment as $k$ is important, because in the case that $m_1$ is divisible by 2, half of the elements of $\bm{k}$ are zero; $k$ and $e_1$ having the same parity would lead to a division by zero above.

Then:\vspace{-0.4cm}
$$\left[\begin{array}{cc}B_0&0\\0&B_1\end{array}\right]    \left[\begin{array}{c}\bm{\bar{w}}_0\\\bm{\tilde w}_1 \end{array}\right]    =     \left[\begin{array}{c}\bm{\bar c}_0\\\bm{\bar c}_1 - 2\frac{\bm{\bar w}_{e_1}}{\bm{k}_{e_1}}\bm{1}^{(k)}  \\\end{array}\right]$$
i.e. these padding widths equalize all segments except for the one pair involving segment $k$.
We choose $k$ such that it is the segment of a hole which is in between the odd-couple pair (cf.~Sec.~\ref{sec:holechain}). Then it is easy (as illustrated in Fig.~\ref{fig:oddcouple}, with local indices) to determine some additional padding values $\pm\delta$ whose addition to $\bm{\tilde w}$ finally yields a vector $\bm{w}^*$ of padding widths which
\begin{itemize}
\item equalize all segments, i.e. $A\bm{w}^* = b$,
\item are zero for both extra cut path segments,
\item but are not generally non-negative.
\end{itemize}

As shown in Sec.~\ref{sec:63}, we can then add a sufficiently large constant padding width to all elements of $\bm{w}^*$ except the two of extra cut path segments to obtain a non-negative solution as required.\qed

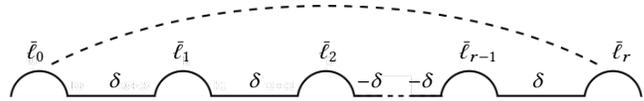
\begin{figure}[tb]
\centering
\begin{overpic}[width=0.99\columnwidth]{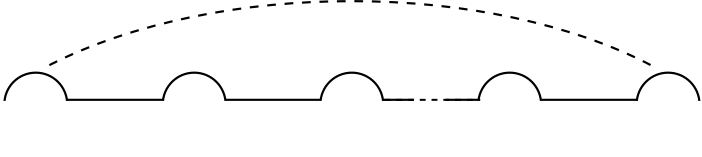}
\put(3,14){\small$\bar\ell_0$}
\put(26,14){\small$\bar\ell_1$}
\put(49,14){\small$\bar\ell_2$}
\put(71,14){\small$\bar\ell_{r-1}$}
\put(95,14){\small$\bar\ell_r$}
\put(16,9.5){\small$\delta$}
\put(38,9.5){\small$\delta$}
\put(55,9.5){\small$-\delta$}
\put(63,9.5){\small$-\delta$}
\put(82.5,9.5){\small$\delta$}
\end{overpic}
\vspace{-0.8cm}
\caption{Final step of equalization. Segments $0$ and $r$ form the odd-couple pair. $\bar\ell$ denotes segment lengths after padding widths $\tilde w$ have been added.
Let $\bar\ell_1$ be the only one length that is not equalized, i.e. it differs by some value $2\delta$ from the padded length of its mate.
By adding $\pm\delta$ as \emph{additional} padding widths as depicted (positive left and right of the unequalized segment, otherwise with alternating sign) we achieve the following:
$\bar\ell_1$ increases by $2\delta$ (which equalizes it), partners $\bar\ell_0$ and $\bar\ell_r$ change by the same value ($+\delta$, due to an odd number of holes in between, preserving their equalization),
all other hole segments keep their length (because their additional padding is $\delta-\delta = 0$).}
\label{fig:oddcouple}    
\vspace{-0.2cm}
\end{figure}

\subsubsection*{Fourfold Case}

In this case the cutgraph contains no extra cut path and we have a single component $M'$ with the number of segments $m$ divisible by 4.
As shown above, the matrix $B$ (which has a single block in this case) does not have full rank in this case.
Its upper left $(m-2) \times (m-2)$-submatrix, however, has full rank (as it is a tridiagonal Toeplitz matrix), thus $B^-$, which is $B$ with the last two rows removed, is a (rectangular) matrix with full row rank.
This implies we can obtain padding widths $\bm{w}$ with $B^- \bm{w} = \bm{1}-\bm{\ell}^-$, i.e. they bring all segments \emph{but the last two} to unit length.

In the case of fourfold cones, the transitions across the cutgraph extension $T$ (cf.~Sec.~\ref{sec:mtop}) are rotations by a multiple of $2\pi$ (= identity); the cut extension $T$ can actually be omitted.
This implies that the boundary of the flattening $F$ is formed exclusively by the segments and is entirely rectilinear, as illustrated in Fig.~\ref{fig:poly}.
Without loss of generality, we assume that all even-index segments are laid out horizontal, thus all odd-index segments vertical.
Walking along $\partial F$, we encounter horizontal segments alternatingly in positive and negative $u$-direction, vertical segments alternatingly in positive and negative $v$-direction.
The fact that $\partial F$ is (and after padding remains) a \emph{closed} polygon then implies
$$\sum_{i = 0,\dots,\nicefrac{m}{2}-1} \!\!\!\!\!\!\!-1^i \ell_{2i} = 0, \quad \quad \sum_{i = 0,\dots,\nicefrac{m}{2}-1} \!\!\!\!\!\!\!-1^i \ell_{2i+1} = 0,$$
which implies that if all even/odd segments but one have unit length, the one has unit length as well. Hence, the last two conditions of $B$ are, in the fourfold case, satisfied automatically if all other conditions are satisfied, thus $B^- \bm{w} = \bm{1} - \bm{\ell}^-$ implies $B \bm{w} = \bm{1}-\bm{\ell}$. Of course $\bm{w}$ may contain negative values, but we can add an arbitrary constant shift $\bm{w}^* = \bm{w} +\lambda\bm{1}$ because $B \bm{w}^* = \mu\bm{1} - \bm{\ell}$ (with $\mu = (2\lambda+1)$), thus $A\bm{w}^*=\bm{b}$ and $\bm{w}^*\geq 0$ for sufficiently large $\lambda$.\qed

\begin{figure}[b]
\centering
\begin{overpic}[width=0.99\columnwidth]{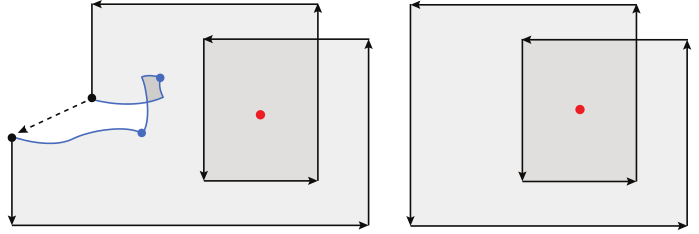}
\put(6,18.0){$\mathbf{d}$}
\put(34,19){\small$-2\pi$}
\put(79.5,19.5){\small$-2\pi$}
\put(20,24.75){\small$-\frac{1}{2}\pi$}
\put(18,10.75){\small$\frac{1}{2}\pi$}

\end{overpic}
\vspace{-0.05cm}
\caption{ Left: boundary $\partial F$ (black) laid out in the plane after cutting to cones (blue). Red indicates a 
cone with $k_i=8$, i.e. curvature $\hat\Theta_i = -2\pi$ (parametric angle $4\pi$) for which a cut is superfluous. 
Right: The segment gap $\mathbf{d}$ vanishes if the cones are fourfold, thus $\partial F$ is a rectilinear polygon.}
\label{fig:poly}    
\end{figure}
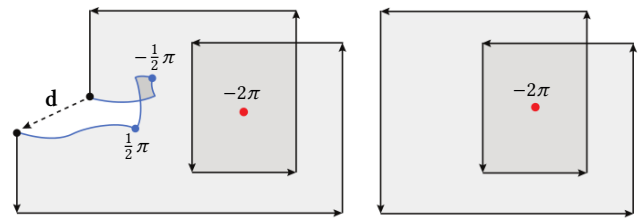

\subsection{Genus 1}
\label{app:g1}

The equalization system for the genus 1 cutgraph pattern is:

\vspace{-0.05cm}
\begin{minipage}[c]{0.4\linewidth}
\begin{overpic}[width=0.85\columnwidth]{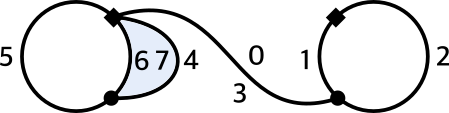}
\end{overpic}
\end{minipage}\quad\;\;
\begin{minipage}[c]{0.4\linewidth}
{\small
\begingroup
\addtolength{\jot}{-0.55em}
\begin{align}
w_{1}+w_{5}-w_{2}-w_{4}&=\ell_{3}-\ell_{0}\notag\\
w_{6}+w_{6}-w_{3}-w_{5}&=\ell_{4}-\ell_{7}\notag\\
w_{0}+w_{2}-w_{7}-w_{7}&=\ell_{6}-\ell_{1}\notag\\
w_{1}+w_{3}-w_{0}-w_{4}&=\ell_{5}-\ell_{2}\notag
\end{align}
\endgroup
}
\end{minipage}
\vspace{-0.2cm}

One can easily verify that the system matrix has full row rank, and that it has positive vectors (e.g. $\bm{1}$) in its kernel, thus has a non-negative solution for any righthand side.

\subsection{Genus 2}
\label{app:g2}
The five different cutgraph patterns covering any possible cone prescription for the genus 2 case are depicted in Fig.~\ref{fig:genus2specialB}.
One can easily verify explicitly that the corresponding system matrices all have full row rank, and that they have positive vectors (e.g. $\bm{1}$) in their kernel, thus have non-negative solutions for any righthand side.

\begin{figure}[h!]
\centering
\vspace{-0.2cm}
\begin{overpic}[width=0.66\columnwidth]{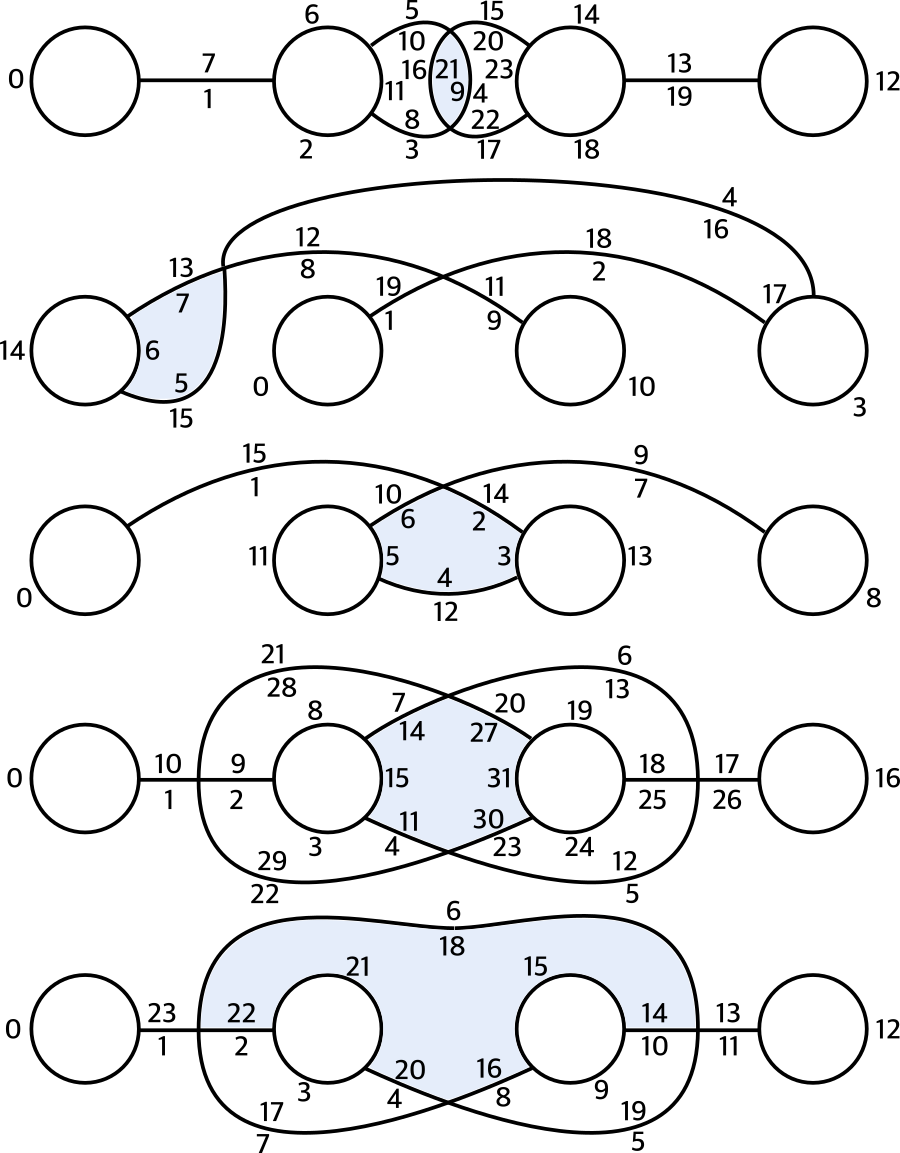}
\end{overpic}
\vspace{-0.25cm}
\caption{Special cutgraph patterns to be used to guarantee equalizability for genus 2 surfaces, depending on whether a subset of cones compatible with a region (shaded) with 2, 3, 5, 6, or 7 corners is present, respectively.}
\label{fig:genus2specialB}    
\vspace{-0.2cm}
\end{figure}  

\hspace{-0.35cm}
\begin{minipage}[t]{0.4\linewidth}
2 corners:
\vspace{-0.1cm}
{\small
\begingroup
\addtolength{\jot}{-0.57em}
\begin{align}
w_{1}+w_{7}-w_{13}-w_{19}&=\ell_{12}-\ell_{0}\notag\\
w_{1}+w_{3}-w_{17}-w_{19}&=\ell_{18}-\ell_{2}\notag\\
w_{5}+w_{7}-w_{13}-w_{15}&=\ell_{14}-\ell_{6}\notag\\
w_{8}+w_{10}-w_{20}-w_{22}&=\ell_{23}-\ell_{11}\notag\\
w_{0}+w_{2}-w_{0}-w_{6}&=\ell_{7}-\ell_{1}\notag\\
w_{2}+w_{17}-w_{11}-w_{16}&=\ell_{8}-\ell_{3}\notag\\
w_{20}+w_{22}-w_{21}-w_{21}&=\ell_{9}-\ell_{4}\notag\\
w_{6}+w_{15}-w_{11}-w_{16}&=\ell_{10}-\ell_{5}\notag\\
w_{12}+w_{14}-w_{12}-w_{18}&=\ell_{19}-\ell_{13}\notag\\
w_{5}+w_{14}-w_{4}-w_{23}&=\ell_{20}-\ell_{15}\notag\\
w_{8}+w_{10}-w_{9}-w_{9}&=\ell_{21}-\ell_{16}\notag\\
w_{3}+w_{18}-w_{4}-w_{23}&=\ell_{22}-\ell_{17}\notag
\end{align}
\endgroup
}
\end{minipage}\quad\;\;
\begin{minipage}[t]{0.4\linewidth}
7 corners:
\vspace{-0.1cm}
{\small
\begingroup
\addtolength{\jot}{-0.57em}
\begin{align}
w_{1}+w_{23}-w_{11}-w_{13}&=\ell_{12}-\ell_{0}\notag\\
w_{2}+w_{4}-w_{8}-w_{10}&=\ell_{9}-\ell_{3}\notag\\
w_{14}+w_{16}-w_{20}-w_{22}&=\ell_{21}-\ell_{15}\notag\\
w_{0}+w_{7}-w_{0}-w_{6}&=\ell_{23}-\ell_{1}\notag\\
w_{3}+w_{17}-w_{18}-w_{21}&=\ell_{22}-\ell_{2}\notag\\
w_{3}+w_{17}-w_{16}-w_{21}&=\ell_{20}-\ell_{4}\notag\\
w_{7}+w_{11}-w_{8}-w_{10}&=\ell_{19}-\ell_{5}\notag\\
w_{13}+w_{23}-w_{14}-w_{22}&=\ell_{18}-\ell_{6}\notag\\
w_{1}+w_{5}-w_{2}-w_{4}&=\ell_{17}-\ell_{7}\notag\\
w_{9}+w_{19}-w_{15}-w_{20}&=\ell_{16}-\ell_{8}\notag\\
w_{9}+w_{19}-w_{15}-w_{18}&=\ell_{14}-\ell_{10}\notag\\
w_{5}+w_{12}-w_{6}-w_{12}&=\ell_{13}-\ell_{11}\notag
\end{align}
\endgroup
}
\end{minipage}

\hspace{-0.35cm}
\begin{minipage}[t]{0.4\linewidth}
\vspace{-0.25cm}
3 corners:
\vspace{-0.1cm}
{\small
\begingroup
\addtolength{\jot}{-0.57em}
\begin{align}
w_{1}+w_{19}-w_{9}-w_{11}&=\ell_{10}-\ell_{0}\notag\\
w_{5}+w_{7}-w_{16}-w_{18}&=\ell_{17}-\ell_{6}\notag\\
w_{2}+w_{4}-w_{13}-w_{15}&=\ell_{14}-\ell_{3}\notag\\
w_{0}+w_{9}-w_{0}-w_{8}&=\ell_{19}-\ell_{1}\notag\\
w_{3}+w_{11}-w_{12}-w_{17}&=\ell_{18}-\ell_{2}\notag\\
w_{3}+w_{13}-w_{12}-w_{17}&=\ell_{16}-\ell_{4}\notag\\
w_{6}+w_{7}-w_{8}-w_{14}&=\ell_{15}-\ell_{5}\notag\\
w_{5}+w_{6}-w_{4}-w_{14}&=\ell_{13}-\ell_{7}\notag\\
w_{15}+w_{19}-w_{16}-w_{18}&=\ell_{12}-\ell_{8}\notag\\
w_{1}+w_{10}-w_{2}-w_{10}&=\ell_{11}-\ell_{9}\notag
\end{align}
\endgroup
}

\vspace{-0.25cm}
5 corners:
\vspace{-0.1cm}
{\small
\begingroup
\addtolength{\jot}{-0.57em}
\begin{align}
w_{1}+w_{15}-w_{7}-w_{9}&=\ell_{8}-\ell_{0}\notag\\
w_{2}+w_{4}-w_{4}-w_{6}&=\ell_{5}-\ell_{3}\notag\\
w_{10}+w_{12}-w_{12}-w_{14}&=\ell_{13}-\ell_{11}\notag\\
w_{0}+w_{10}-w_{0}-w_{9}&=\ell_{15}-\ell_{1}\notag\\
w_{3}+w_{6}-w_{7}-w_{13}&=\ell_{14}-\ell_{2}\notag\\
w_{3}+w_{5}-w_{11}-w_{13}&=\ell_{12}-\ell_{4}\notag\\
w_{2}+w_{5}-w_{1}-w_{11}&=\ell_{10}-\ell_{6}\notag\\
w_{8}+w_{14}-w_{8}-w_{15}&=\ell_{9}-\ell_{7}\notag
\end{align}
\vspace{-0.35cm}
\endgroup
}
\end{minipage}\quad\;\;
\begin{minipage}[t]{0.4\linewidth}
\vspace{-0.25cm}
6 corners:
\vspace{-0.1cm}
{\small
\begingroup
\addtolength{\jot}{-0.57em}
\begin{align}
w_{1}+w_{10}-w_{17}-w_{26}&=\ell_{16}-\ell_{0}\notag\\
w_{2}+w_{4}-w_{23}-w_{25}&=\ell_{24}-\ell_{3}\notag\\
w_{7}+w_{9}-w_{18}-w_{20}&=\ell_{19}-\ell_{8}\notag\\
w_{11}+w_{14}-w_{27}-w_{30}&=\ell_{31}-\ell_{15}\notag\\
w_{0}+w_{22}-w_{0}-w_{21}&=\ell_{10}-\ell_{1}\notag\\
w_{3}+w_{29}-w_{8}-w_{28}&=\ell_{9}-\ell_{2}\notag\\
w_{3}+w_{29}-w_{15}-w_{30}&=\ell_{11}-\ell_{4}\notag\\
w_{22}+w_{26}-w_{23}-w_{25}&=\ell_{12}-\ell_{5}\notag\\
w_{17}+w_{21}-w_{18}-w_{20}&=\ell_{13}-\ell_{6}\notag\\
w_{8}+w_{28}-w_{15}-w_{27}&=\ell_{14}-\ell_{7}\notag\\
w_{6}+w_{16}-w_{5}-w_{16}&=\ell_{26}-\ell_{17}\notag\\
w_{13}+w_{19}-w_{12}-w_{24}&=\ell_{25}-\ell_{18}\notag\\
w_{13}+w_{19}-w_{14}-w_{31}&=\ell_{27}-\ell_{20}\notag\\
w_{6}+w_{10}-w_{7}-w_{9}&=\ell_{28}-\ell_{21}\notag\\
w_{1}+w_{5}-w_{2}-w_{4}&=\ell_{29}-\ell_{22}\notag\\
w_{12}+w_{24}-w_{11}-w_{31}&=\ell_{30}-\ell_{23}\notag
\end{align}
\endgroup
}
\end{minipage}

\section{Cone Metric Existence: Proof of Proposition~\ref{prop:conemetric}}
\label{sec:troyaproof}

 \cite{Troyanov:1991} presents a general proof of cone metric existence on closed surfaces. It ``extends [\dots] to surfaces with (piecewise geodesic) boundary'', but the piecewise geodesic boundary case, which is the case relevant for our cone metric with rectilinear boundary, is not explicitly spelled out in detail. \cite{Cherrier:1984} focuses on the case with boundary, but does not specifically consider the here relevant delta distributions of curvature. We thus provide a proof tailored to this special case.

Let $M'$ be one of the disk-topology connected components of the cut surface (we will drop the index of the component in the following).  
Consider the expansion $M'_{\text{exp}}$ of  $M'$, obtained by joining a copy of the geodesic disk of size $\epsilon_p$ in $M$ centered at $\pi(p)$, to each boundary point of $M'$.
Multiple $p \in \partial M'$ corresponding to the same $\pi(p)$ get separate copies of the disk centered at $\pi(p)$ and $\epsilon_p$ is chosen sufficiently small for each $p$ so that $M'_\text{exp}$ still has disk-topology. 

To simplify the exposition, we assume that on the surface $M$ the branches of the cut form right angles -- the proof can be extended to arbitrary angles, as long as the curves are transversal, but requires a more complex
solution $\phi_1$ below, with additional cones at the corners, as explained in more detail in \cite{Bunin:2008}. 

Consider a conformal map $f$ from $M'_{\text{exp}}$  to the plane (e.g. to a disk).
As $M'_{\text{exp}}$ has disk-topology, such a map always exists. As $M'$ is in the interior of $M'_{\text{exp}}$ the
map is conformal at the points of the boundary $\partial M'$. The conformal scale factor $|f'|$, where $f'$ is the complex derivative of the map expressed in local complex coordinates on the tangent plane, defines the conformal metric on $\text{Int}(f(M'_{\text{exp}}))$, in particular, on all of $f(M') = M''$ including
the boundary. Let $\gamma_i$, $i=1\ldots m$, be the curves of the boundary of $M''$; these curves are smooth, as the boundary of $M'$ is smooth, and meet at right angles. 

We now construct on  $M''$ a metric with the desired properties; then the metric on $M'$ is obtained by a pullback through $f$. As $M''$ is flat, the equation for the metric in the interior points $x$ of $M$ simplifies to
$$\Delta \phi = \sum_j   \hat\Theta_j \delta(f(p_j)-x),$$
where $\hat\Theta_j$ is the target curvature at cone $c_j = (p_j, \hat\Theta_j)$.
If the geodesic curvature at non-corner boundary points is given by a smooth function $\kappa$, then we have the Neumann boundary condition
$$\frac{\partial \phi}{\partial n} = -\kappa,$$
that needs to be satisfied to obtain straight boundary edges in the final metric.  Note that $\kappa$ may be  discontinuous at the corner points but it is still in $L_2$.  
We can find a particular solution $u_1$ satisfying the Poisson equation on $M''$  without boundary conditions directly as $\phi_1 = \sum_j \hat\Theta_j \ln(|z-f(p_j)|)$ with singularities at $f(p_j)$. 

Then we solve the Laplace equation  $\Delta \phi_2 = 0$, for $\phi_2$ with smooth Neumann conditions $\nicefrac{\partial \phi_2}{\partial n} = \kappa -\nicefrac{\partial\phi_1}{\partial n}$. For this problem to have a solution, the Neumann boundary condition needs to integrate to zero
over the boundary. Observe that because the domain $M''$ is flat, the integral of the geodesic curvature
$\kappa$ over the boundary, with the sum of corner angles $n \frac{\pi}{2}$ added, must be $2\pi$, i.e.,
$\int_{\partial \Omega} \kappa ds = 2\pi -n\frac{\pi}{2}$.  In addition, $\int_{\partial\Omega} \nicefrac{\partial \phi_2}{\partial n} ds
= \sum_j \hat\Theta_j$, by the Gauss theorem.  
Finally, note that by the cutgraph admissibility assumption
on the number of corners, $2\pi - n\frac{\pi}{2}  - \sum_j \hat\Theta_j = 0$, i.e., the integral condition for the Neumann
problem is satisfied.   Therefore, the problem has a unique, up to a constant, solution. This solution is in $H^2$ (and, by Sobolev Lemma,  $C^0$solution up to the boundary) for domains with piecewise smooth boundary and convex corners between curves (cf. \cite[p.~174]{grisvard1985elliptic}). The sum $\phi = \phi_1+\phi_2$ satisfies the Poisson equation and boundary conditions.  The metric $\phi$ is nonsingular at the boundary, therefore it is conformal,  and the angles between boundary curves are preserved. We conclude that the pullback of this metric to $M'$ is the needed metric.

\section{Map Padding}
\label{sec:padding}

As laid out in Sec.~\ref{sec:pad1}, map padding consists of the application of stretch maps to rectangular regions, and lateral shifts within these. We define these operations precisely in the following.

To this end, in the plane with Euclidean coordinates $u,v$, we, w.l.o.g., consider the case of a horizontal segment $s_j$ (aligned with the $u$-axis) to be padded by $w_j$ in upward (positive $v$) direction, as illustrated in Fig.~\ref{fig:padrect} -- the other cases (horizontal downward, vertical left and right) are handled analogously.

In the case of a segment split by $T$, we assume that $T$ (which can be chosen freely) meets the segment at a right angle with a straight cut in the parametric domain. Then both parts can simply be treated separately using the following operations without any special case handling -- except for the same rectangle thickness being used for both parts.

\paragraph{Stretching}
Let $\tau_j$ be the thickness (here: the height) of $R_j$, and $(u_{j \min}, v_{j \min})$ the coordinates of the lower left corner of $R_j$. The map $g_i$ applied to the strip to perform the stretching is a simple one-dimensional scaling by factor $o_j = \frac{w_j+\tau_j}{\tau_j}$:
\begin{equation}
g_j : (u,v) \mapsto (u, v_{j\min} + o_j(v-v_{j \min}))
\end{equation}

\paragraph{Shifting}

We apply a deformation (lateral shifting) within a rectangle $R_j^p$ that leads to a (piecewise) constant speed parametrization of the segment $s_j$.
We use a simple blend (linear in $v$) between the map $\phi_j: [u_{j\min},u_{j\max}] \rightarrow [u_{j\min},u_{j\max}]$ that reparametrizes segment $s_j$ to (piecewise) constant speed (applied at the top of the strip) and the identity map $u\mapsto u$ (applied at the bottom):
\begin{equation}
r_j : (u,v) \mapsto (t \phi_j(u) + (1- t) u, v),
\end{equation}

where $t = (v-v_{j\min})/\tau_j$ is the normalized relative $v$-coordinate within $R_j^p$.
$\phi_j$ is a constant speed reparametrizaton for simple segments. For complex segments it is with piecewise constant speed, constant per boundary curve the segment consists of, such that the lengths of these boundary curves after reparametrization are in the same ratio as the padded lengths of their mates.

One easily verifies that this deformation $r_j$ is injective: the determinant of its Jacobian is $\det J(u,v) = \left(\nicefrac{\partial \phi_j}{\partial u}(u) - 1\right)t +1$, and due to $0 \leq t \leq 1$ and $\nicefrac{\partial \phi_j}{\partial u}(u) > 0$ (as the scaled arc-length reparametrization is non-degenerate and orientation preserving) it is always positive.

\begin{proof}[Proof of Proposition~\ref{prop:rotseam2}]
 $F$ is rotationally seamless, in particular locally injective and continuous (on $M^c$). If $F^{p,m}$ is continuous, so is $f^{(p,m+1)}$ because $g_{m+1}$ is continuous and it is identity on the interface between $S_{m+1}$ and the rest of $M^c$.
If $F^{s,m}$ is continuous, so is $F^{(s,m+1)}$ because $r_{m+1}$ is continuous and it is identity on the interface between $S_{m+1}$ and the rest of $M^c$.
It follows that $F^s$ is continuous.
Analogously, as $g_j$ and $r_j$ are injective, local injectivity is preserved for $F^s$.
Both types of maps,  $g_j$ and $r_j$, preserve the straightness and the orientation of all segments and therefore the pairwise angles between them, thus $F^s$ is rotationally seamless like $F$.
As angles between boundary curve images are not affected, cone angles are preserved as well.
Each boundary curve which segment $s_j$ consists of is parametrized with constant speed by $F^{s,j}$ by construction. As $s_k$ with $k \neq j$ is identity on $s_j$ (more precisely: that part of $s_j$ contained in $R_k$ and thus potentially affected by $s_k$), $F^s(s_j) = F^{s,j}(s_j)$.
\end{proof}

\section{Illustrative Example }
\label{sec:illexample}

We consider the simplest example: a torus with two cones, $k_0 = 2$, $k_1 = 6$, i.e. cone angles
 $\pi$ and $3\pi$, as depicted in Figure~\ref{fig:torus-example}.

\begin{figure}[h]
\centering
\begin{overpic}[width=0.99\columnwidth]{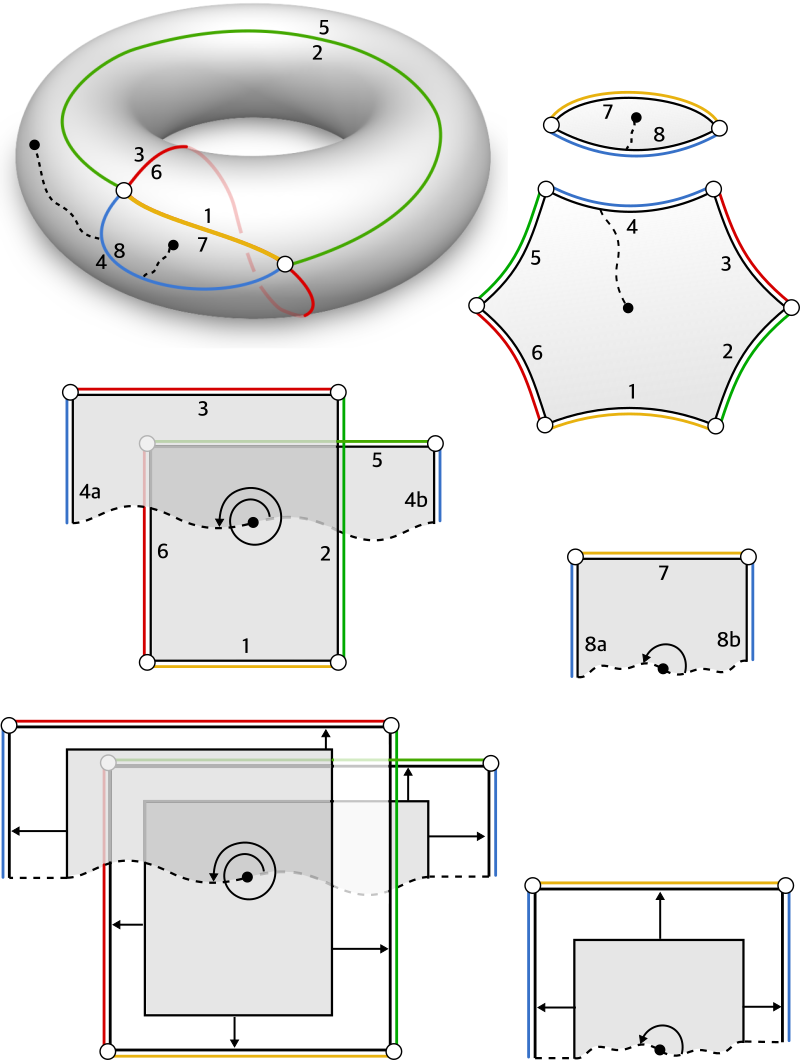}
\put(2,22.5){\small$w_4$}
\put(41.6,22.1){\small$w_4$}
\put(32.6,11.5){\small$w_2$}
\put(22.6,2.2){\small$w_1$}
\put(10.5,13.8){\small$w_6$}
\put(39.0,25.5){\small$w_5$}
\put(31.4,29.5){\small$w_3$}

\put(62.7,13.3){\small$w_7$}
\put(50.8,6.1){\small$w_8$}
\put(70.4,6.2){\small$w_8$}

\end{overpic}
\vspace{-0.05cm}
\caption{
Top left: genus 1 surface with cutgraph consisting of 4 branches (yellow, green, red, blue).
The cutgraph cuts the surface into two components with 2 and 6 corners, respectively, i.e. with a total of 8 boundary segments (two corresponding to each branch).
Top right: schematic depiction of the two components under a cone metric with rectilinear boundary consisting of straight segments (here shown schematically as curved arcs) meeting at right angles. 
Middle left/right: planar flattening of the two components implied by the metric (after cutting to cones -- dashed).
The numbering of segments is used to set up the system for padding widths $w_i$.
Bottom left/right: the padded flattening (padding, indicated by arrows, in white).}
\label{fig:torus-example}    
\vspace{-0.25cm}
\end{figure}
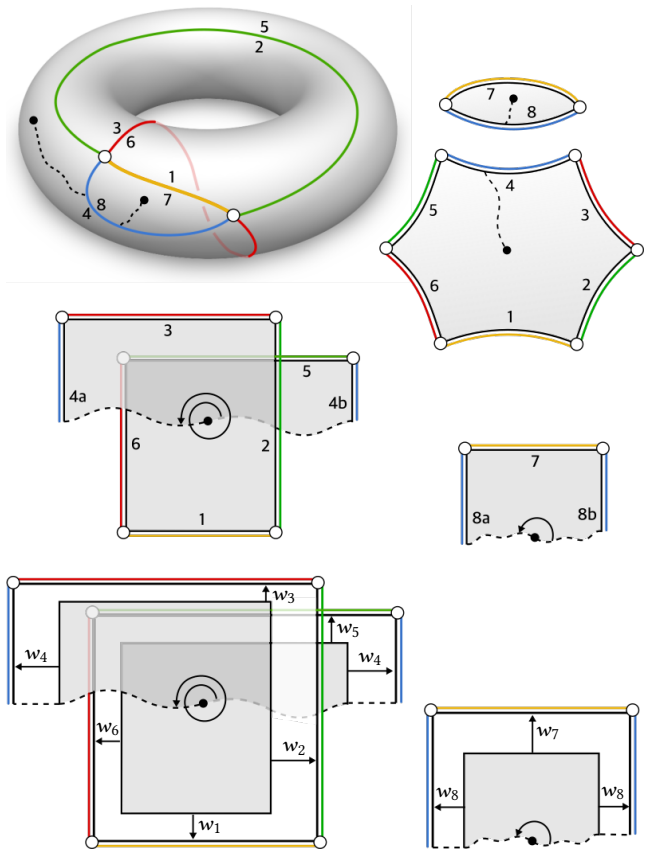  
 
 \paragraph{Cutgraph} We cut the surface into a 2-corner and a 6-corner component (cf.~Fig.~\ref{fig:torus-example} top).
 The cutgraph was embedded in the surface in such a way that the total cone curvature contained in each component is compatible with the number of corners in terms of Gauss-Bonnet: a cone with $k_0=2$ lies in the 2-corner region, a cone with $k_1=6$ in the 6-corner region.

\paragraph{Cone Metric}
We compute 
 a cone metric on each of the two components (e.g. conformal, given by a pointwise scale factor) which is flat everywhere except at the cones, where it has prescribed curvature.
  In addition, we require the boundary to be geodesically straight
 at all boundary points except for the corners, where it forms right angles under the metric.  
  
\paragraph{Metric to Parametrization}
If we add further cuts connecting all cones to the boundary (indicated with dashed curves) this cone metric is flat in the entire interior and thus yields a flattening, a global parametrization of the torus, with two charts (cf.~Fig.~\ref{fig:torus-example} middle left and right).
The image of each of the two maps is a domain with (away from the dashed cone cut) rectilinear boundary: straight segments meeting at right angles. 
Consequently, as the angle between any two segments is some integer multiple of $\nicefrac{\pi}{2}$, this parametrization is \emph{rotationally} seamless, but it may have a jump in scale across cuts. In particular, two segments corresponding to the same cutgraph branch -- here (1,7), (2,5), (3,6), and (4,8) -- may have different lengths in general.

\paragraph{Equalization by Padding}
In order to equalize the lengths of identified pairs of segments (and thereby enable the parametrization to ultimately become seamless), we add \emph{padding}, i.e. we extend the parametric domain by shifting straight segments in orthogonal direction (cf.~Fig.~\ref{fig:torus-example} bottom).
For each segment $i$, numbered sequentially around each component, $\ell_i$, $i= 1 \ldots 8$, is its parametric length. 
  For a segment $i$, after padding its length becomes $\ell_i + w_{\text{prev}(i)} + w_{\text{next}(i)}$, where 
$\text{prev}(i)$ and $\text{next}(i)$ are previous and next segment indices around the component, and $w_j$ is the padding width for segment~$j$.  Equating the post-padding lengths of all four pairs of identified segments yields the following four equations in this example:
\begin{equation}
\begin{split}
&\ell_1 + w_2 + w_6 = \ell_7 + 2 w_8;\quad \ell_2 + w_1 + w_3 = \ell_5 + w_4 + w_6;\\
&\ell_3 + w_2 + w_4 = \ell_6 + w_1 + w_5;\quad \ell_4 + w_3 + w_5 = \ell_8 + 2w_7
\end{split}
\nonumber
\label{eq:torus26}
\end{equation}
where $\ell_i$ are the known segment lengths, $w_i$ are the unknown padding widths. 
The matrix of this equation system has the form
\begin{equation}
\nonumber
A = \left[\begin{array}{rrrrrrrr} 
 0 &1 &0 &0  &0  &1  &0  &-2\\
 1 &0 &1 &-1 &0  &-1 &0  &0\\
-1 &1 &0 &1  &-1 &0  &0  &0\\ 
 0 &0 &1 &0  &1  &0  &-2 &0\\
\end{array}\right]
\end{equation}
and the right-hand side is $\bm{b} = [\ell_7-\ell_1,\, \ell_5-\ell_2,\, \ell_6-\ell_3,\, \ell_8-\ell_4]^T$, i.e., parametric length mismatches of identified segments.  
To  \emph{equalize} segments lengths, we need to find a solution of the system $A\bm{w} = \bm{b}$, where $\bm{w}$ is the vector of padding 
widths, such that $\bm{w}\geq 0$. This non-negativity condition is important to guarantee that the domain does not degenerate through padding. 
Observe that $A$ has full (row) rank, thus admits a solution. Observe further that $A \mathbf{1} = 0$ in this case, i.e., after computing an arbitrary solution, we can obtain a non-negative solution by adding a sufficiently large constant.
More generally, note
that $A$ (thus its rank and its nullspace) is determined solely by the choice of cutgraph combinatorics. 
For instance, without the blue or without the yellow branch, the cutgraph (cutting the surface to a single topological disk in these cases) would yield a system that does \emph{not} have a non-negative solution $\bm{w}$ for every possible $\bm{b}$.

\paragraph{Seamless Parametrization} Once the padded domain is obtained, we remap the original image onto this domain. This is done by stretching outwards thin strips running along the segments to cover the added space in the rectangular regions padded onto the domain, yielding a seamless global parametrization.

\end{document}